\documentclass[runningheads]{llncs}
\usepackage[utf8]{inputenc}
\usepackage{graphicx}
\usepackage{hyperref}
\usepackage{amssymb}
\usepackage{amsmath}
\usepackage{mathtools}
\usepackage[linesnumbered,ruled,vlined]{algorithm2e}
\usepackage{listings}
\usepackage{booktabs}
\usepackage{multirow}
\usepackage[table,xcdraw]{xcolor}
\usepackage{wrapfig}

\newcommand{\HRule}{\rule{\linewidth}{0.2mm}}

\def\accd{\mathsf{nsb}}
\def\acce{\mathsf{nsb_e}}
\def\prec{\mathsf{prec}}

\def\ufp{\mathsf{ufp}}
\def\ulp{\mathsf{ulp}}
\def\ufpe{\mathsf{ufp_e}}
\def\ulpe{\mathsf{ulp_e}}
\def\rr{\mathbb R}
\def\const{\mathsf{Const}}
\def\lw{\varrho}
\def\Min{\operatorname*{Min}}
\def\st{\operatorname*{s.t.}}
\def\nn{\mathbb N}
\def\x{\mathbf x}

\newtheorem{lem}{Lemma}

 \newcommand\ForAuthors[1]
 {\par\smallskip                     
  \begin{center}
   \fbox
   {\parbox{0.9\linewidth}
    {\raggedright\sc--- #1}
   }
  \end{center}
  \par\smallskip                     %
 }

 \title{Fast and Efficient Bit-Level Precision Tuning}
 \author{}
 \institute{}
 \author{Assal\'e Adj\'e\inst{1} \and Dorra Ben Khalifa\inst{1}\and  Matthieu Martel\inst{1,2}}
\authorrunning{Ben Khalifa et al.}\institute{University of Perpignan, LAMPS laboratory, 52 Av. P. Alduy, Perpignan, France \and Numalis, Cap Omega, Rond-point Benjamin Franklin, Montpellier, France
 \email{\{assale.adje,matthieu.martel,dorra.ben-khalifa\}@univ-perp.fr}}
\begin{document}
\maketitle
\begin{abstract}
In this article, we introduce a new technique for precision tuning. This problem
consists of finding the least data types for numerical values such that the result
of the computation satisfies some accuracy requirement. State of the art techniques
for precision tuning use a try and fail approach. They change the data types of some variables
of the program and evaluate the accuracy of the result. Depending on what is obtained,
they change more or less data types and repeat the process. Our technique is radically different.
Based on semantic equations, we generate an Integer Linear Problem (ILP) from the program source
code. Basically, this is done by reasoning on the most significant bit and the number of
significant bits of the values which are integer quantities. The integer solution to this problem, computed in polynomial time by a (real) linear programming solver, gives the optimal data types at the bit level.
A finer set of semantic equations is also proposed which does not reduce directly
to an ILP problem. So we use policy iteration to find the solution.
Both techniques have been implemented and we show that our results encompass the results of
state of the art tools.

\smallskip
 \keywords{Static analysis, computer arithmetic,
integer linear problems, numerical accuracy, policy iteration.}
\end{abstract}

\section{Introduction}\label{sec1}

Let us consider a program $P$ computing some numerical result $R$, typically but not necessarily in the IEEE754
floating-point arithmetic \cite{IEEE754}. Precision tuning then consists of finding the smallest
data types for all the variables and expressions of $P$ such that the result $R$ has some
desired accuracy. These last years, much attention has been paid to this problem \cite{CBBSGR17,DHS18,GR18,KSWLB19,LHSL13,RGNNDKSBIH13}. Indeed, precision tuning makes it possible to save memory and, by way of consequence, it has a positive impact on the footprint of programs concerning energy consumption, bandwidth usage, computation time, etc.

A common point to all the techniques cited previously is that they follow a try and fail strategy.
Roughly speaking, one chooses a subset $S$ of the variables of $P$, assigns to them  smaller data
types (e.g. \texttt{binary32} instead of \texttt{binary64} \cite{IEEE754}) and evaluates the
accuracy of the tuned program $P'$. If the accuracy of the result returned by $P'$ is satisfying
then new variables are included in $S$ or even smaller data types are assigned to certain
variables already in $S$ (e.g. \texttt{binary16}). Otherwise, if the accuracy of
the result of $P'$ is not satisfying, then some variables are removed from $S$. This process is
applied repeatedly, until a stable state is found. Existing techniques differ in their way to
evaluate the accuracy of programs, done by dynamic analysis
\cite{GR18,KSWLB19,LHSL13,RGNNDKSBIH13} or by static  analysis \cite{CBBSGR17,DHS18} of $P$ and $P'$.
They may also differ in the algorithm used to define $S$, delta debugging being
the most widespread method \cite{RGNNDKSBIH13}.

Anyway all these techniques suffer from the same
combinatorial limitation: If $P$ has $n$ variables and if the method tries $k$ different data types
then the search space contains $k^n$ configurations. They scale neither in the number $n$ of
 variables (even if heuristics such as delta debugging \cite{RGNNDKSBIH13} or branch and bound \cite{CBBSGR17} reduce the search space at the price of
optimality) or in the number $k$ of data types which can be tried. In particular, bit level
precision tuning, which consists of finding the minimal number of bits needed for each variable
to reach the desired accuracy, independently of a limited number $k$ of data types, is not an option.

So the method introduced in this article for precision tuning of programs is radically different.
Here, no try and fail method is employed. Instead, the accuracy of the arithmetic expressions assigned
to variables is determined by semantic equations, in function of the accuracy of the operands. By
reasoning on the number of significant bits of the variables of $P$ and knowing the weight of
their most significant bit thanks to a range analysis performed before the tuning phase, we
are able to reduce the problem to an Integer Linear Problem (ILP) which can be optimally solved
in one shot by a classical linear programming solver (no iteration). Concerning the number $n$ of variables, the
method scales up to the solver limitations and the solutions are naturally found at the bit level,
making the parameter $k$ irrelevant. An important point is that the optimal solution to the continuous linear programming relaxation of our ILP is a vector of integers,
as demonstrated in Section \ref{sec52}. By consequence, we may use a linear solver among real numbers
whose complexity is polynomial~\cite{schrijver1998theory} (contrarily to the linear solvers among integers whose complexity is NP-complete~\cite{papadimitriou1981complexity}).
This makes our precision tuning method solvable in polynomial-time, contrarily to the existing exponential methods.

Next, we go one step further by introducing a second set of semantic equations.
These new equations make it possible to tune even more the precision by being less pessimistic
on the propagation of carries in arithmetic operations. However the problem do not reduce
any longer to an ILP problem ($\min$ and $\max$ operators are needed). Then we use
policy iteration (PI) \cite{CGGMP05} to find efficiently the solution.

Both methods have been implemented inside a tool for precision tuning named XXX\footnote{In this article, XXX hides the actual name of our tool and missing references refer to our previous work for anonymity.} \cite{KMA19}. Formerly, XXX was expressing the precision tuning problem as a set of first order logical
 propositions among relations between linear integer expressions. An SMT solver (Z3 in practice
 \cite{MB08}) was used repeatedly to find the existence of a solution with a certain weight
 expressing the number of significant bits ($\accd$) of the variables \cite{KMA19}. In the present article, we compare
experimentally our new methods to the SMT based method
previously used by XXX \cite{KMA19} and to the Precimonious tool \cite{GR18,RGNNDKSBIH13}. These experiments on programs coming from mathematical libraries or other applicative domains such as
IoT \cite{KM19} show that the technique introduced in this article for precision tuning clearly
encompasses the state of the art techniques.

The rest of this article is organized as follows. In the next section, we provide a motivating example. We then present in Section \ref{sec4} some essential background on the functions needed for the constraint generation and also we detail the set of constraints for both ILP and PI methods. Section \ref{sec5} presents the proofs of correctness. We end up in Section \ref{sec6} by showing that our new techniques exhibits very good results in practice before concluding in Section \ref{sec7}. 
\section{Motivating Example}\label{sec2}

A motivating example to better explain our method is given by the code snippet of Figure \ref{pendulum}. In this example, we aim at modeling the movement of a simple pendulum without damping. Let $l = 0.5 \ m$ be the length of this pendulum, $m = 1 \ kg$ its mass and $g=9.81\ m \cdot s^{-2}$ Newton's gravitational constant. We denote by $\theta $ the tilt angle in radians as shown in Figure \ref{pendulum} (initially $\theta = \frac{\pi}{4}$). The Equation describing the movement of the pendulum is given in Equation (\ref{pend1}).

 \lstset{language=Java, numbers=left, showspaces=false,
   showstringspaces=false, tabsize=2, breaklines=true,
   xleftmargin=5.0ex,
    numberstyle=\scriptsize,numbersep=0pt
}
\begin{figure}[tb]
\hrule
\vspace{0.2cm}
\scriptsize
\centering
\begin{tabular}{lcl}\tt
\begin{lstlisting}[mathescape]
 g = 9.81; l = 0.5;
 y1 = 0.785398; y2 = 0.785398;
 h = 0.1; t = 0.0;
 while (t<10.0) {
  y1new = y1 + y2 * h ;
  aux1 = sin(y1) ;
  aux2 = aux1 * h * g / l;
  y2new = y2 - aux2;
  t = t + h;
  y1 = y1new; y2 = y2new;
 };
 require_nsb(y2,20);
  \end{lstlisting}
&~\hspace{0.4cm}~&
\begin{minipage}{4cm}
~\\ \vspace{2.0cm} ~  \includegraphics[width=4.0cm]{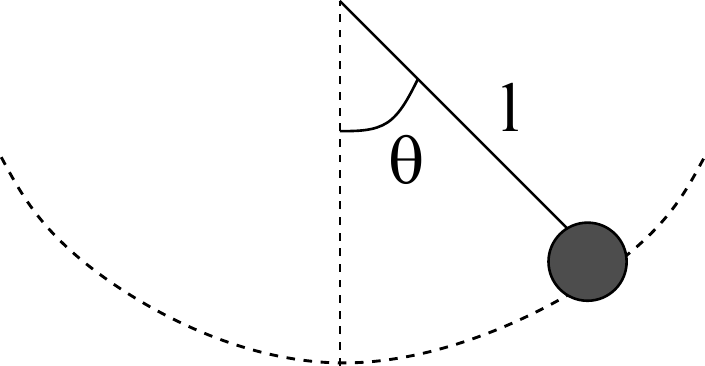}\label{tilt} \end{minipage}\\
&&\\
\tt
  \begin{lstlisting}[mathescape]
 g$^{\ell_1}$ = 9.81$^{\ell_0}$; l$^{\ell_3}$ = 0.5$^{\ell_2}$;
 y1$^{\ell_5}$ = 0.785398$^{\ell_4}$;
 y2$^{\ell_7}$ = 0.785398$^{\ell_6}$;
 h$^{\ell_9}$ = 0.1$^{\ell_8}$; t$^{\ell_{11}}$ = 0.0$^{\ell_{10}}$;
 while (t$^{\ell_{13}}$ <$^{\ell_{15}}$ 10.0$^{\ell_{14}}$)$^{\ell_{59}}$ {
  y1new$^{\ell_{24}}$ = y1$^{\ell_{17}}$ +$^{\ell_{23}}$ y2$^{\ell_{19}}$ *$^{\ell_{22}}$ h$^{\ell_{21}}$;
  aux1$^{\ell_{28}}$ = sin(y1$^{\ell_{26}}$)$^{\ell_{27}}$;
  aux2$^{\ell_{40}}$ = aux1$^{\ell_{30}}$ *$^{\ell_{39}}$ h$^{\ell_{32}}$
   *$^{\ell_{38}}$ g$^{\ell_{34}}$ /$^{\ell_{37}}$ l$^{\ell_{36}}$;
  y2new$^{\ell_{46}}$ = y2$^{\ell_{42}}$ -$^{\ell_{45}}$ aux2$^{\ell_{44}}$;
  t$^{\ell_{52}}$ = t$^{\ell_{48}}$ +$^{\ell_{51}}$ h$^{\ell_{50}}$;
  y1$^{\ell_{55}}$ = y1new$^{\ell_{54}}$;
  y2$^{\ell_{58}}$ = y2new$^{\ell_{57}}$;
 };
 require_nsb(y2,20)$^{\ell_{61}}$;
  \end{lstlisting}
&~\hspace{0.4cm}~&
\tt
  \begin{lstlisting}[mathescape]
 g|20| = 9.81|20|; l|20| = 1.5|20|;
 y1|29| = 0.785398|29|;
 y2|21| = 0.0|21|;
 h|21| = 0.1|21|; t|21| = 0.0|21|;
 while (t<1.0) {
   y1new|20| = y1|21| +|20| y2|21|
     *|22| h|21|;
   aux1|20| = sin(y1|29|)|20|;
   aux2|20| = aux1|19| *|20| h|18|
     *|19| g|17| /|18|l|17|;
   y2new|20| = y2|21| -|20| aux2|18|;
   t|20| = t|21| +|20| h|17|;
   y1|20| = y1new|20|;
   y2|20| = y2new|20|;
 };
 require_nsb(y2,20);
  \end{lstlisting}
\end{tabular}
\vspace{0.2cm}
\hrule
\caption{\label{pendulum} Top left: source program. Top right: pendulum movement for $\theta = \frac{\pi}{4}$. Bottom left: program annotated with labels. Bottom right: program with inferred accuracies.}
\end{figure}

\begin{equation}\label{pend1}\small
m\cdot l\cdot \frac{d^2 \theta}{dt^2}= - m \cdot g\cdot \sin \theta
\end{equation}

Equation (\ref{pend1}) being a second order differential equation, we need to transform it into a system of two first order differential equations for resolution so we obtain $y_1 =\theta$ and  $y_2=\frac{d\theta}{dt}$. By applying Euler's method to these last equations, we obtain Equation (\ref{pend3}) implemented in Figure \ref{pendulum}.
\begin{equation}\label{pend3}\small
\frac{dy_1}{dt}  =  y_2\quad \text{and} \quad
\frac{dy_2}{dt}  =  -\frac{g}{l}\cdot \sin y_1
\end{equation}

The key point of our technique is to generate a set of constraints for each statement of our imperative language introduced further in Section \ref{sec4}. For our example, we suppose that all variables, before XXX analysis, are in double precision (source program in the top left corner of Figure \ref{pendulum}) and that a range determination is performed by dynamic analysis on the program variables (we plan to use a static
analyzer in the future). XXX assigns to each node of the program's syntactic tree a unique control point in order to determine easily the number of significant bits of the result as mentioned in the bottom corner of Figure \ref{pendulum}. Some notations can be stressed about the structure of XXX source code. For instance, the annotation \texttt{g}$^{\ell_1} = 9.81^{\ell_0}$ denotes that \texttt{g} has a unique control point $\ell_{1}$. As well, we have the statement \textbf{\texttt{require\_nsb(y2,}}20\texttt{)} which informs the tool that the user wants to get on variable \texttt{y2} only 20 bits (we consider that a result has $n$ significants if the relative error between the exact and approximated results is less than $2^{-n}$). Finally, the minimal precision needed for the inputs and intermediary results satisfying the user assertion is observed on the bottom right corner of Figure~\ref{pendulum}. In this code, if we consider for instance Line $6$, then \texttt{y1new|20|} means that the variable needs $20$ significant bits at this point. Similarly, \texttt{y1} and \texttt{y2} need $21$ bits each and the addition requires $20$ bits.

The originality of our method is that we reduce the precision tuning problem to an ILP. For example, taking again Line $6$ of the pendulum code, we generate six constraints as shown in Equation \ref{ilpcstr} (this is detailed further in Section \ref{sec4}).
  \begin{equation}\label{ilpcstr}
\scriptsize
C_1=
\left\{
\begin{array}{l}
 \accd(\ell_{17}) \ge \accd(\ell_{23})+ (-1)+ \xi(\ell_{23})(\ell_{17}, \ell_{22}) - (-1), \\
 \accd(\ell_{22} \ge \accd(\ell_{23}+0+\xi(\ell_{23})(\ell_{17}, \ell_{22})- (1), \\
 \accd(\ell_{19}) \ge \accd(\ell_{22})+\xi(\ell_{22})(\ell_{19}, \ell_{21})-1, \\
 \accd(\ell_{21}) \ge \accd(\ell_{22})+\xi(\ell_{22})(\ell_{19}, \ell_{21})-1,\\
 \accd(\ell_{23}) \ge\accd(\ell_{24}), \ \xi(\ell_{23}) (\ell_{17}, \ell_{22})\ge1, \ \ \xi(\ell_{22})(\ell_{19}, \ell_{21})\ge1 \\
\end{array}\right\}
\end{equation}

The first two constraints are for the addition. Basically, $\accd(\ell_{23})$ stands for number of significant bits as described in Section \ref{sec31}. It represents the difference between the unit in the first place ($\ufp$, see Section \ref{sec31}) of the result of the sum and the $\ufp$ of its error denoted $\ufpe(\ell_{23})$. We have $\ufp(\ell_{17}) = \ufp(\ell_{17}) -\accd(\ell_{17})$. As mentioned previously, the $\ufp$ are computed by a prior range analysis. Then, at constraint generation time, they are constants. For our example, $\ufp(\ell_{17}) = -1$. This quantity occurs in the first constraints. The next two constraints are for the multiplication. The fourth constraint $\accd(\ell_{23}) \ge \accd(\ell_{24})$ is for the assignment and the last two constraints are for the constant functions $\xi(\ell_{23})(\ell_{17}, \ell_{22})$ and $\xi(\ell_{22})(\ell_{19}, \ell_{21})$ respectively for the addition and multiplication (see next paragraphs and Section \ref{sec4} for more details).
Note that XXX generates such constraints for all the statements of the program.

 For a user requirement of 20 bits on the variable \texttt{y2} as shown in the original program of the top right corner of Figure \ref{pendulum} (all variables are in double precision initially), XXX succeeds in tuning the majority of variables of the pendulum program into simple precision with a total number of bits at bit level equivalent to 274 (originally the program used 689 bits). The new mixed precision formats obtained for line $6$ are \texttt{y1new|20|} \texttt{=} \texttt{y1|21|}\ \texttt{+|20|} \ \texttt{y2}\texttt{|22|}\ $\times$\texttt{|22|}\ \texttt{h|22|}.

Let us now focus on the term $\xi(\ell_{23})(\ell_{17}, \ell_{22})$ (for the addition). In our ILP, we always assume that $\xi$ is a constant function equal to $1$. This corresponds to the carry bit which can be propagated up to the $\ufp$ and $\ufpe$ and increments them by $1$ which is correct but pessimistic. In large codes, this function becomes very costly if we perform several computations
at a time and therefore the errors would be considerable, especially that in many cases adding this carry bit is useless because the operands and their errors do not overlap. It is then crucial to use the most precise function $\xi$.
Unfortunately, when we model this optimization the problem is no more linear ($\min$ and $\max$ operators arise) and we have to use the PI technique \cite{CGGMP05} to solve it (see Section \ref{sec42}).

In this case, by analyzing Line $6$ of our program, we have to add the following new constraints (along with the former ones) as mentioned in Equation (\ref{ipcstr}). In fact, policy iteration makes it possible to break the $\min$ in the $\xi(\ell_{23})(\ell_{17}, \ell_{22})$ function by choosing the $\max$ between $\ufp(\ell_{22}) - \ufp(\ell_{17}) - \accd(\ell_{17}) - \accd(\ell_{22}) - \acce(\ell_{17})$ and $0$, the $\max$ between $\ufp(\ell_{17}) - \ufp(\ell_{22}) + \accd(\ell_{22}) - \accd(\ell_{17}) - \acce(\ell_{22})$ and $0$ and the constant $1$. Next, it becomes possible to solve the corresponding ILP. If no fixed point is reached, XXX iterates until a solution is found as shown in Section \ref{sec42}. By applying this optimization, the new formats of line $6$ are given as
\texttt{y1new|20|} \texttt{=} \texttt{y1|21|}\  \texttt{+|20|}\  \texttt{y2|21|}\  $\times$\texttt{|22|}\  \texttt{h|21|}. By comparing with the formats obtained above, a gain of precision of $1$ bit is observed on variables \texttt{y2} and \texttt{h} (total of 272 bits at bit level for the optimized program). The program on the bottom right corner of Figure~\ref{pendulum} illustrates the new optimized formats obtained by using the policy iteration technique.
  \begin{equation}\label{ipcstr}
\scriptsize
C_2=
\left\{
\begin{array}{l}
 \acce(\ell_{23}) \ge  \acce(\ell_{17}), \\
 \acce(\ell_{23}) \ge \acce(\ell_{22}), \\
 \accd(\ell_{23}) \ge -1-0+ \accd(\ell_{22})- \accd(\ell_{17})+ \acce(\ell_{22})+\xi(\ell_{23}, \ell_{17}, \ell_{22}), \\
 \acce(\ell_{23}) \ge 0-(-1)+ \accd(\ell_{17})- \accd(\ell_{22})+ \acce(\ell_{17})+\xi(\ell_{23}, \ell_{17}, \ell_{22}), \\
 \acce(\ell_{23}) \ge \acce(\ell_{24}),\\
 \acce(\ell_{22}) \ge \accd(\ell_{19}) + \acce(\ell_{19})+ \acce(\ell_{21})-2, \\
 \acce(\ell_{22}) \ge \accd(\ell_{21}) + \acce(\ell_{21})+ \acce(\ell_{19})-2, \\
  \xi(\ell_{23})(\ell_{17}, \ell_{22}) = \min\left(\begin{array}{l}\max\big(0 -6 + \accd(\ell_{17}) - \accd(\ell_{22}) - \acce(\ell_{17}), 0\big),\\
                                                         \max\big(6 - 0 + \accd(\ell_{22}) -\accd(\ell_{17}) - \acce(\ell_{22}), 0\big),1 \end{array}\right) \\
\end{array}\right\}
\end{equation}


\begin{figure}[tb]
\small
\noindent\rule{12.4cm}{0.25mm}
\begin{equation}\label{cst}
\varepsilon(c^{\ell}\texttt{\#p}) \leq  2^{\ufp(c) - \min\big(\texttt{p}, \prec(\ell)\big)} 
\end{equation}
\begin{equation}\label{add}
\varepsilon(c_1^{\ell_1}\texttt{\#$p_1$} +^\ell c_2^{\ell_2}\texttt{\#$p_2$}) \leq \varepsilon({c_1^{\ell_1}\texttt{\#$p_1$}}) +
\varepsilon({c_2^{\ell_2}\texttt{\#$p_2$}}) + {2^{\ufp(c_1 + c_2)- \prec(\ell)}}
\end{equation}
\begin{equation}\label{sub}
\varepsilon(c_1^{\ell_1}\texttt{\#$p_1$} -^\ell c_2^{\ell_2}\texttt{\#$p_2$}) \leq \varepsilon({c_1^{\ell_1}\texttt{\#$p_1$}}) -
\varepsilon({c_2^{\ell_2}\texttt{\#$p_2$}}) + {2^{\ufp(c_1 -  c_2)- \prec(\ell)}}
\end{equation}
\begin{equation}\label{times}
\begin{array}{c}
\varepsilon(c_1^{\ell_1}\texttt{\#$p_1$} \times^\ell c_2^{\ell_2}\texttt{\#$p_2$}) \leq \\ c_1 \cdot \varepsilon({c_2^{\ell_2}\texttt{\#$p_2$}}) +
  c_2 \cdot \varepsilon({c_1^{\ell_1}\texttt{\#$p_1$}}) + \varepsilon({c_1^{\ell_1}\texttt{\#$p_1$}}) \cdot \varepsilon({c_2^{\ell_2}\texttt{\#$p_2$}})
 + {2^{\ufp(c_1 \times c_2)- \prec(\ell)}}
 \end{array}
\end{equation}
\begin{equation}\label{divide}
\varepsilon(c_1^{\ell_1}\texttt{\#$p_1$} \div^\ell c_2^{\ell_2}\texttt{\#$p_2$}) \leq
 \varepsilon(c_1^{\ell_1}\texttt{\#$p_1$} \times^\ell c_2^{\prime\ell_2}\texttt{\#$p_2$})
  \quad \text{with} \quad c^\prime_2 = \frac{1}{c_2}
\end{equation}
\begin{equation} \label{maths}
  \varepsilon\left(\phi(c^{\ell_1}\texttt{\#p})^\ell\right) \leq 2^{\ufp(\phi(c)) - \texttt{p} + \varphi} +  2^{\ufp(\phi(c)) - \prec(\ell)} \ \text{with} \quad \phi \in \{\sin, \cos, \tan, \log, \ldots\}
\end{equation}
\begin{equation}\label{racine}
  \varepsilon\left(\sqrt{(c^{\ell_1}\texttt{\#p})}^\ell\right) \leq 2^{\ufp(\sqrt{c}) - \texttt{p}} + 2^{\ufp(\sqrt{c}) - \prec(\ell)}
\end{equation}
\noindent\rule{12.4cm}{0.25mm}
 \caption{Numerical error on arithmetic expressions.}\label{error}
\end{figure}

\section{Constraints Generation for Bit-Level Precision Tuning}\label{sec4}

In this section, we start by providing essential definitions for understanding the rest of the article. Also, we define a simple imperative language (see Figure~\ref{imp}) from which we generate semantic equations in order to determine the least precision needed for the program numerical values. Then, we will focus on the difference between the two sets of constraints obtained when using the simple ILP and the more complex PI formulations which optimizes the carry bit that can propagate throughout computations. The operational semantics of the language as well as the theorem proving that the solution to the system of constraints gives the desired $\accd$ when running programs are detailed further in Section \ref{sec5}.
\subsection{Elements of Computer Arithmetic}\label{sec31}
Our technique is independent of a particular computer arithmetic. In fact, we manipulate numbers for which we know their unit in the first place ($\ufp$) and the number of significant digits ($\accd$) defined as follows.
\begin{description}
  \item[Unit in the First Place] The unit in the first place of a real number $x$ (possibly encoded up to some rounding mode by a floating-point or a fixpoint number) is given in Equation~(\ref{ufp}). This function, which is independent of the representation of $x$, will be used further in this section to describe the error propagation across the computations.
  \begin{equation}\small
     \label{ufp}
       \ufp(x) = \min \{i \in \mathbb{Z} : 2^{i+1} > x \} = \lfloor  \log_2(x) \rfloor \enspace.
     \end{equation}
  \item[Number of Significant Bits] Intuitively, $\accd(x)$ is the number of significant bits of $x$. Formally, following Parker \cite{parker}, if $\accd(x) = k$, for $x\not=0$ then the error $\varepsilon(x)$ on $x$ is less than $2^{\ufp(x) - k}$. If $x=0$ then $\accd(x) = 0$. For example, if the exact binary value $1.0101$ is approximated by either $x=1.010$ or $x=1.011$ then $\accd(x) = 3$.
\end{description}
In the following, we also use $\ufpe(x)$ and $\acce(x)$ to denote the $\ufp$ and $\accd$ of the error on $x$, i.e. $\ufpe(x) = \ufp(\varepsilon(x))$ and  $\acce(x) = \accd(\varepsilon(x))$.

In this article, we consider a finite precision arithmetic, independently of any particular representation (IEEE754 \cite{IEEE754}, POSIT \cite{CGGMP05}, $\ldots$). Nevertheless, the representation being finite, roundoff errors may arise when representing values or performing elementary operation. These errors are defined in Figure \ref{error}. First of all, by definition, using the function $\ufp$ of Equation (\ref{ufp}), for a number $x$ with $\texttt{p}$ number of significant bits, the roundoff error $\epsilon(x)$ is bounded as shown in Equation (\ref{bound}).
\begin{equation}\label{bound}\small
\varepsilon(x) \le 2^{\ufp(x) - \texttt{p} + 1}
\end{equation}
Let $\prec(\ell)$ be the precision of the operation at control point $\ell$. For example, the precision is $53$ bits for the IEEE754 \texttt{binary64} format. In fact, $\prec(\ell)$ is used to compute the truncation error of the operations. Equations~(\ref{cst}) to (\ref{racine}) of Figure~\ref{error} define the numerical errors of the arithmetic expressions of our language (presented in Section~\ref{sec4}). For constants occurring in the code, the initial precision must be given by the user and we write $c^{\ell}\texttt{\#p}$ a constant $c$ with $\texttt{p}$ significant bits at control point $\ell$. Then, following Equation~(\ref{bound}), $\varepsilon(c^{\ell}\texttt{\#p})$ is defined in Equation~(\ref{cst}) of Figure~\ref{error}: the $\accd$ of the constant is $\min(\texttt{p},\prec(\ell))$ and consequently the error is bounded by $ 2^{\ufp(c) - \min\big(\texttt{\texttt{p}}, \prec(\ell)\big)}$. In equations~(\ref{add}) to (\ref{divide}), we propagate the errors on the operands and we add the roundoff error due to the elementary operation itself. For an elementary function $\phi$, we assume that $\varphi$ bits are lost as shown in Equation~(\ref{maths}) (more details are given in Section \ref{sec4}). The last equation is for the square root function. This function being computable exactly, no more error than for the elementary operation is introduced.

\subsection{Integer Linear Problem Formulation}\label{sec41}

First, we define in Figure \ref{imp} the simple imperative language in which our input programs are written.
  \begin{figure*}[tb]
\scriptsize
\noindent\rule{12.4cm}{0.25mm}
         \begin{center} $\ell \in Lab$    \quad  $x \in Id$ \quad $\odot$ $\in$ \{+, -, $\times$, $\div$\} \quad $math$ $\in$ \{$\sin$, $\cos$, $\tan$, $\arcsin$, $\log$, $\ldots$\}  \end{center}
 \textbf{Expr $\ni$ e} : e ::= c\texttt{\#$p^\ell$} $|$ $x^\ell$ $|$ $e_1^{\ell_1} \odot ^\ell e_2^{\ell_2}$ $|$ $math(e^{\ell_1})^\ell$  $|$ $sqrt(e^{\ell_1})^\ell$ \newline
 \newline
\textbf{Cmd $\ni$ c} : $c ::= c_1^{\ell_1} ; c_2^{\ell_2}$ $|$ $x =^\ell e^{\ell_1}$ $|$  $\textbf{\textit{while}}^\ell \: b^{\ell_0} \: \textbf{\textit{do}}\:  \: c_1^{\ell_1}$ $|$ $\textbf{\textit{if}}^\ell \: b^{\ell_0} \: \textbf{\textit{then}} \:\: c_1^{\ell_1}  \:\textit{\textbf{else}}  \: \: c$ $|$ $\textbf{\textit{require\_nsb}}(x,n)^\ell$
\noindent\rule{12.4cm}{0.25mm}
 \caption{\label{imp} Simple imperative language of constraints.}
\end{figure*}

\begin{figure}[!h]\small
\HRule\vspace{0.2cm}
$$
\mathcal{E}[ c\texttt{\#p}^\ell] \lw= \emptyset \quad (\textsc{Const})
\hspace{1cm}
\mathcal{E}[x^\ell] \lw = \big\{ \accd(\lw(x))\geq \accd(\ell) \big\} \quad (\textsc{Id})
$$
\vspace{0.1cm}
$$
\begin{array}{c}
\mathcal{E}[e_1^{\ell_1}+^{\ell}e_2^{\ell_2}] \lw = \mathcal{E}[e_1^{\ell_1}] \lw \ \cup \ \mathcal{E}[e_2^{\ell_2}] \lw
\\
 \cup\\
 \left\{\accd(\ell_1) \geq \accd(\ell) + \ufp(\ell_1) - \ufp(\ell) +  \xi(\ell)(\ell_1, \ell_2),\right.\\
\left.\hspace{0.1cm}\accd(\ell_2) \geq \accd(\ell) + \ufp(\ell_2) - \ufp(\ell) + \xi(\ell)(\ell_1, \ell_2)
 \right\}
 \end{array}\quad (\textsc{Add})
$$
\vspace{0.1cm}
$$
\begin{array}{c}
\mathcal{E}[e_1^{\ell_1}-^{\ell}e_2^{\ell_2}] \lw = \mathcal{E}[e_1^{\ell_1}] \lw \ \cup \ \mathcal{E}[e_2^{\ell_2}] \lw
\\
 \cup\\
 \left\{\accd(\ell_1) \geq \accd(\ell) + \ufp(\ell_1) - \ufp(\ell) +  \xi(\ell)(\ell_1, \ell_2) ,\right.\\
\left.\hspace{0.1cm}\accd(\ell_2) \geq \accd(\ell) + \ufp(\ell_2) - \ufp(\ell) + \xi(\ell)(\ell_1, \ell_2)
 \right\}
 \end{array} \quad (\textsc{Sub})
$$
\vspace{0.1cm}
$$
\begin{array}{c}
\mathcal{E}[e_1^{\ell_1}\times^{\ell}e_2^{\ell_2}] \lw = \mathcal{E}[e_1^{\ell_1}] \lw \ \cup \ \mathcal{E}[e_2^{\ell_2}] \lw
\\
 \cup\\
 \left\{\accd(\ell_1) \geq \accd(\ell) + \xi(\ell)(\ell_1, \ell_2) - 1 ,\right.
\left. \hspace{0.1cm}\accd(\ell_2) \geq \accd(\ell) + \xi(\ell)(\ell_1, \ell_2) - 1
 \right\}
 \end{array}\quad (\textsc{Mult})
 $$
 \vspace{0.1cm}
$$
\begin{array}{c}
\mathcal{E}[e_1^{\ell_1}\div^{\ell}e_2^{\ell_2}] \lw = \mathcal{E}[e_1^{\ell_1}] \lw \ \cup \ \mathcal{E}[e_2^{\ell_2}] \lw
\\
 \cup\\
 \left\{\accd(\ell_1) \geq  \accd(\ell) + \xi(\ell)(\ell_1, \ell_2) - 1 ,\right.
\left. \hspace{0.1cm}\accd(\ell_2) \geq \accd(\ell) + \xi(\ell)(\ell_1, \ell_2) - 1
 \right\}
 \end{array}\quad (\textsc{Div})
 $$
 $$
\mathcal{E}\left[ \sqrt{e^{\ell_{1}}}^{\ell} \right]  \lw= \mathcal{E}[e_1^{\ell_1}] \lw \   \cup \   \big\{\accd(\ell_1) \geq \accd(\ell) \big\} \quad (\textsc{Sqrt})
$$
 $$
\mathcal{E}\left[ \phi\big({e^{\ell_{1}}} \big)^{\ell}\right]  \lw= \mathcal{E}[e_1^{\ell_1}] \lw \  \cup \ \big\{\accd(\ell_1) \geq \accd(\ell) + \varphi \big\}\ \text{with}\ \phi \in \{\sin, \cos, \tan, \log, \ldots\}
 \quad  (\textsc{Math})
$$
 $$\begin{array}{c}
\mathcal{C}\left[ x \texttt{:=}^\ell e^{\ell_1}\right]  \lw= \big(C, \lw \left[x \mapsto \ell\right] \big) \
 \text{where}\ C =  \mathcal{E}[e_1^{\ell_1}] \lw  \cup \left\{\accd(\ell_1) \geq \accd(\ell) \right\}
\end{array} \quad (\textsc{Assign})
$$

 $$\begin{array}{c}
\mathcal{C}\left[c_1^{\ell_1} \texttt{;} c_2^{\ell_2}  \right]  \lw= \big(C_1 \cup C_2 , \lw_2\big)

\\ \text{where}\ \big(C_1, \lw_1 \big) = \mathcal{C}\left[c_1^{\ell_1}\right]\lw \ \text{and}\ \big(C_2, \lw_2 \big) = \mathcal{C}\left[c_2^{\ell_2}\right]\lw_1
\end{array} \quad (\textsc{Seq})
$$
\begin{equation*}
\begin{array}{c}
\mathcal{C}[\texttt{if}^\ell\ e^{\ell_0}\ \texttt{then}\ c^{\ell_1}\ \texttt{else}\ c^{\ell_2}]\ \lw =
 (C_1\cup C_2\cup C_3,\lw')\ \\ \text{where}\ \left|
\begin{array}{l}
 \forall x\in\text{Id},\ \lw'(x)=\ell,\
(C_1,\lw_1)=\mathcal{C}[c_1^{\ell_1}]\ \lw,\
(C_2,\lw_2)=\mathcal{C}[c_2^{\ell_2}]\ \lw,\\
C_3= \underset{x\in \text{Id}}{\bigcup} \left\{
 \accd(\lw_1(x)) \geq \accd(\ell), \
  \accd(\lw_2(x)) \geq \accd(\ell)
\right\}
\end{array}\right.\end{array}\quad (\textsc{Cond})
\end{equation*}
\begin{equation*}
\begin{array}{c}
\mathcal{C}[\texttt{while}^\ell\ e^{\ell_0}\ \texttt{do}\ c^{\ell_1}]\ \lw =
 (C_1\cup C_2,\lw')\ \\ \text{where}\ \left|
\begin{array}{l}
 \forall x\in\text{Id},\ \lw'(x)=\ell, \
(C_1,\lw_1)=\mathcal{C}[c_1^{\ell_1}]\ \lw'\\
C_2= \underset{x\in \text{Id}}{\bigcup} \left\{
 \accd(\lw(x)) \geq \accd(\ell),\
  \accd(\lw_1(x)) \geq \accd(\ell)
\right\}
\end{array}\right.\end{array}\quad (\textsc{While})
\end{equation*}
$$
\mathcal{C}[\texttt{require\_nsb}(x, \texttt{p})^\ell ] \lw = \big\{\accd(\lw(x)) \geq \texttt{p} \big\} \quad (\textsc{Req})
$$
\center\noindent\rule{10cm}{0.2mm}\\
\vspace{0.2cm}
$\xi(\ell)(\ell_1, \ell_2) = 1$
\vspace{0.3cm}
\HRule
\caption{\label{ilp}ILP constraints with pessimistic carry bit propagation $\xi=1$.}
\end{figure}

 We denote by $Id$ the set of identifiers and by $Lab$ the set of control points of the program as a means to assign to each element $e \in Expr$ and $c \in Cmd$ of our language a unique control point $\ell$ $\in$ $Lab$. First, in $c\texttt{\#p}$, $\texttt{p}$ indicates the number of significant bits of the constant $c$ in the source code. The parameter $\texttt{p}$ is computed by our tool XXX, when solving the constraints.
 Next, the statement \texttt{require\_nsb}(x,n)$^\ell$ indicates the number of significant bits $n$ that a variable $x$ must have at
a control point $\ell$. The rest of the grammar is standard.

As we have mentioned, we are able to reduce the problem of determining the lowest precision on variables
 and intermediary values in programs to an Integer Linear Problem (ILP) by reasoning on their unit in the first place ($\ufp$) and the number of significant bits.
In addition, we assign to each control point $\ell$ an integer variable $\accd(\ell)$ corresponding to the $\accd$ of the arithmetic expressions. $\accd(\ell)$  is determined by solving the ILP generated by the rules of Figure \ref{ilp}.

 Let us now focus on the rules of Figure \ref{ilp} where  $\lw\ :\ \text{Id}\rightarrow \text{Id}\times \text{Lab}$ is an environment which relates each identifier $x$ to its last assignment $x^\ell$: Assuming that $x\ \texttt{:=}^\ell e^{\ell_1}$ is the last assignment of $x$, the environment $\lw$ maps $x$ to $x^\ell$. Then, $\mathcal{E}[e]\ \lw$ generates the set of constraints for an expression $e \in Expr$ in the environment $\lw$.
  We now formally define these constraints for each element of our language. No constraint is generated for a constant $c\texttt{\#p}$ as mentioned in Rule (\textsc{Const}) of Figure \ref{ilp}. For Rule (\textsc{Id}) of a variable $x^\ell$, we require that the $\accd$ at control point $\ell$ is less than its $\accd$ in the last assignment of $x$ given in $\lw(x)$. For a binary operator $\odot$ $\in$ \{+, -, $\times$, $\div$\}, we first generate the set of constraints $\mathcal{E}[e_1^{\ell_1}] \lw$ and $\mathcal{E}[e_2^{\ell_2}] \lw$ for the operands at control points $\ell_1$ and $\ell_2$. Considering Rule (\textsc{ADD}), the result of the addition of two numbers is stored in control point $\ell$. Recall that a range determination is performed before the accuracy analysis, $\ufp(\ell)$, $\ufp(\ell_1)$ and $\ufp(\ell_2)$ are known at constraint generation time.

   Now, before going further in the explanation of the constraints generation for binary operations, we introduce the function $\xi$ which computes the carry bit that can occur throughout an addition (similar reasoning will be done for the other elementary operations).
  In the present ILP of Figure~\ref{ilp}, we over-approximate the function $\xi$ by $\xi(\ell)(\ell_1,\ell_2)=1$ for all $\ell,\ \ell_1$ and $\ell_2$, thus assuming the worst case, i.e. a carry bit is  added at each operation. We will optimize $\xi$ in Section \ref{sec42} but the problem will not remain an ILP any longer. To wrap up, for the addition (Rule (\textsc{Add})), the $\accd(\ell)$ of the exact result is the number of bits between $\ufp(\ell_1 + \ell_2)$ and the $\ufp$ of the error $e$ which is:
  \begin{equation}\label{tunisienne}\small
  e = \max\big(\ufp(\ell_1) - \accd(\ell_1), \ufp(\ell_2) - \accd(\ell_2) \big) - \xi(\ell)(\ell_1, \ell_2)
  \end{equation}
  Hence, the error on the addition in precision $\prec(\ell)$ is
  \begin{equation}\label{tunisienneprime}\small
  e^\prime = \max\big(\ufp(\ell_1) - \accd(\ell_1), \ufp(\ell_2) - \accd(\ell_2), \prec(\ell) \big) - \xi(\ell)(\ell_1, \ell_2)
  \end{equation}
  Equation (\ref{tunisienne}) is obtained from Equation (\ref{add}):
 \small \begin{eqnarray*}\small
\varepsilon(c_1^{\ell_1}\texttt{\#$p_1$} +^\ell c_2^{\ell_2}\texttt{\#$p_2$}) &\leq& \varepsilon({c_1^{\ell_1}\texttt{\#$p_1$}}) +
\varepsilon({c_2^{\ell_2}\texttt{\#$p_2$}}) + {2^{\ufp(c_1 + c_2)- \prec(\ell)}}\\
&=&\max \big( \varepsilon({c_1^{\ell_1}\texttt{\#$p_1$}}) ,
\varepsilon({c_2^{\ell_2}\texttt{\#$p_2$}}) \big)+ {2^{\ufp(c_1 + c_2)- \prec(\ell)}} - \xi(\ell)(\ell_1, \ell_2)\\
&=&\max \big( \ufp(\ell_1)-\accd(\ell_1) ,
\ufp(\ell_1)-\accd(\ell_1) \big)\\
&&+ {2^{\ufp(c_1 + c_2)- \prec(\ell)}} - \xi(\ell)(\ell_1, \ell_2)
  \end{eqnarray*}\normalsize
   Since $\accd(\ell) \leq \prec(\ell)$, we may get rid of the last term of $e^\prime$ in Equation~(\ref{tunisienneprime}) and the two equations generated for Rule (\textsc{ADD}) are derived from Equation~(\ref{accd}).
  \begin{equation}\label{accd}\small
 \accd(\ell) \leq \ufp(\ell) - \max\big(\ufp(\ell_1) - \accd(\ell_1), \ufp(\ell_2) - \accd(\ell_2) \big) - \xi(\ell)(\ell_1, \ell_2)
\end{equation}
For example, let us consider the piece of code hereafter.

{\centering\begin{minipage}{6cm}
{\tt \small\begin{lstlisting}[mathescape]
 x$^{\ell_2}$ = 5.0$^{\ell_1}$; y$^{\ell_4}$ = 3.0$^{\ell_3}$;
 z$^{\ell_8}$ = x$^{\ell_6}$ +$^{\ell_5}$ y$^{\ell_7}$;
 require_nsb(z,15)$^{\ell_{9}}$;
  \end{lstlisting}\normalsize}\end{minipage}}

 Wrapping up our constraints, we have
  $\
  15 \leq \accd(\ell_9) \leq \accd(\ell_8) \leq \accd(\ell_5),
 $ and
 $
 \accd(\ell_6) \geq \accd(\ell_8) + \ufp(\ell_6) -\ufp(\ell_7) + \xi(\ell_8)(\ell_6, \ell_7)
 $.
 Since $\ufp(\ell_6) = 2$ and $\ufp(\ell_7) = 1$ we have
  $
 \accd(\ell_6) \geq \accd(\ell_8) + 1 + \xi(\ell_8)(\ell_6, \ell_7)
 $
  and, consequently,
    $
 \accd(\ell_6) \geq 15 + 1 + 1 = 17
 $. Rule (\textsc{Sub}) for the subtraction is obtained similarly to the addition case.
 For Rule (\textsc{Mult}) of multiplication (and in the same manner Rule(\textsc{Div})), the reasoning mimics the one of the addition. Let $x$ and $y$ be two floating point numbers and $z$ the result of their product, $z = x^{\ell_1} \times^{\ell} y^{\ell_2}$. We denote by $\varepsilon(x)$, $\varepsilon(y)$ and $\varepsilon(z)$ the errors on $x$, $y$ and $z$, respectively. The error $\varepsilon(z)$ of this multiplication is $\varepsilon(z) = x \cdot \varepsilon(y) + y \cdot \varepsilon(x) + \varepsilon(x) \cdot \varepsilon(y)$. These numbers are bounded as shown in Equation~(\ref{product}).
 \begin{equation}\label{product} \small
 \begin{array}{rcl}
2^{\ufp(x)} \le x \le 2^{\ufp(x) + 1} \quad  \text{and} \quad 2^{\ufp(x) - \accd(x)} & \le & \varepsilon(x) \le 2^{\ufp(x) - \accd(x)+ 1}  \\
2^{\ufp(y)} \le y \le 2^{\ufp(y) + 1} \quad  \text{and} \quad 2^{\ufp(y) - \accd(y)} & \le & \varepsilon(y) \le 2^{\ufp(y) - \accd(y) + 1}  \\
2^{\ufp(x) + \ufp(y) - \accd(y)} + 2^{\ufp(y) + \ufp(x) - \accd(x)} & \le & \varepsilon(z) \le 2^{\ufp(z) - \accd(z) + 1}\\
+  2^{\ufp(x) + \ufp(y) - \accd(x) - \accd(y)} &&
\end{array}
 \end{equation}
By getting rid of the last term $2^{\ufp(x) + \ufp(y) - \accd(x) - \accd(y)}$ of the error $\epsilon(z)$ which is strictly less than the former two ones, assuming that $\ufp(x+y) = \ufp(z)$ and, finally, by reasoning on the exponents, we obtain the equations of Rule (\textsc{Mult}).
\[\small\accd(x) \ge \accd(z) + \xi(\ell)(\ell_1, \ell_2) - 1 \quad  \text{and} \quad \accd(y) \ge \accd(y) + \xi(\ell)(\ell_1, \ell_2) - 1 \]

 Although the square root (\textsc{Sqrt}) is included, e.g, in the IEEE754 Standard, it is not the case for the other elementary functions such as the natural logarithm, the exponential functions and the hyperbolic and trigonometric functions gathered in Rule (\textsc{Math}). Also, each implementation of these functions has its own $\accd$ which we have to know to model the propagation of errors in our analyses. To cope with this limitation, we consider that each elementary function
introduces a loss of precision of $\varphi$ bits, where $\varphi \in \mathbb{N}$ is a parameter of the analysis and consequently of our tool, XXX.

The rules of commands are rather classical, we use control points to distinguish many assignments of the same variable and also to implement joins in conditions and loops. Given a command $c$ and an environment $\lw$, $\mathcal{C}[c]\ \lw$ returns a pair $(C,\lw')$ made of a set $C$ of constraints and of a new environment $\lw'$. The function $\mathcal{C}$ is defined by induction on the structure of commands in figures~\ref{ilp} and ~\ref{ip}.
 For conditionals, we generate the constraints for the \texttt{then} and \texttt{else} branches plus additional constraints to join the results of both branches. For loops, we relate the number of significants bits at the end of the \texttt{body} to the $\accd$ of the same variables and the beginning of the loop.
\begin{figure}[!h]\small
\HRule
\vspace{0.2cm}
$$
\mathcal{E}^\prime[ c\texttt{\#p}^\ell] \lw= \big\{\acce(\ell) = 0  \big\}  \quad (\textsc{Const}^\prime)
\hspace{1cm}
\mathcal{E}^\prime[x^\ell] \lw = \big\{\acce(\lw(x))\geq \acce(\ell) \big\} \quad (\textsc{Id}^\prime)
$$
\vspace{0.1cm}
$$
\begin{array}{c}
\mathcal{E}^\prime[e_1^{\ell_1}+^{\ell}e_2^{\ell_2}] \lw = \mathcal{E}^\prime[e_1^{\ell_1}] \lw \ \cup \ \mathcal{E}^\prime[e_2^{\ell_2}]\lw \quad (\textsc{Add}^\prime)
\\
 \cup\\
 \left\{ \begin{array}{c}\acce(\ell) \geq \acce(\ell_1), \ \acce(\ell) \geq \acce(\ell_2), \\
  \acce(\ell) \geq \ufp(\ell_1) - \ufp(\ell_2) +\accd(\ell_2) -\accd(\ell_1) + \acce(\ell_2) + \xi(\ell)(\ell_1, \ell_2), \\
  \acce(\ell) \geq \ufp(\ell_2) - \ufp(\ell_1) +\accd(\ell_1) -\accd(\ell_2) + \acce(\ell_1) + \xi(\ell)(\ell_1, \ell_2)
  \end{array} \right\}
 \end{array}
$$
$$
\begin{array}{c}
\mathcal{E}^\prime[e_1^{\ell_1}-^{\ell}e_2^{\ell_2}] \lw = \mathcal{E}^\prime[e_1^{\ell_1}] \lw \ \cup \ \mathcal{E}^\prime[e_2^{\ell_2}]\lw \quad(\textsc{Sub}^\prime)
\\
 \cup\\
 \left\{ \begin{array}{c}\acce(\ell) \geq \acce(\ell_1), \ \acce(\ell) \geq \acce(\ell_2), \\
  \acce(\ell) \geq \ufp(\ell_1) - \ufp(\ell_2) +\accd(\ell_2) -\accd(\ell_1) + \acce(\ell_2) + \xi(\ell)(\ell_1, \ell_2), \\
  \acce(\ell) \geq \ufp(\ell_2) - \ufp(\ell_1) +\accd(\ell_1) -\accd(\ell_2) + \acce(\ell_1) + \xi(\ell)(\ell_1, \ell_2)
  \end{array} \right\}
 \end{array}
$$
$$
\begin{array}{c}
\mathcal{E}^\prime[e_1^{\ell_1}\times^{\ell}e_2^{\ell_2}] \lw = \mathcal{E}^\prime[e_1^{\ell_1}] \lw \ \cup \ \mathcal{E}^\prime[e_2^{\ell_2}]\lw \quad (\textsc{Mult}^\prime)
\\
 \cup\\
 \left\{ \begin{array}{c}
   \acce(\ell) \geq \accd(\ell_1) + \acce(\ell_1) + \acce(\ell_2) -  2, \
   \acce(\ell) \geq \accd(\ell_2) + \acce(\ell_2) + \acce(\ell_1) -  2
  \end{array} \right\}
 \end{array}
$$
$$
\begin{array}{c}
\mathcal{E}^\prime[e_1^{\ell_1}\div^{\ell}e_2^{\ell_2}] \lw = \mathcal{E}^\prime[e_1^{\ell_1}] \lw \ \cup \ \mathcal{E}^\prime[e_2^{\ell_2}]\lw \quad (\textsc{Div}^\prime)
\\
  \cup\\
 \left\{ \begin{array}{c}
   \acce(\ell) \geq \accd(\ell_1) + \acce(\ell_1) + \acce(\ell_2) -  2, \
   \acce(\ell) \geq \accd(\ell_2) + \acce(\ell_2) + \acce(\ell_1) -  2
  \end{array} \right\}
 \end{array}
$$
 $$
\mathcal{E}^\prime\left[ \sqrt{e^{\ell_{1}}}^{\ell} \right]  \lw= \mathcal{E}^\prime[e_1^{\ell_1}] \lw \   \cup \  \big\{\acce(\ell) \geq \acce(\ell_1) \big\} \quad (\textsc{Sqrt}^\prime)
$$
 $$
\mathcal{E}^\prime\left[ \phi\big({e^{\ell_{1}}} \big)^{\ell}\right]  \lw= \mathcal{E}^\prime[e_1^{\ell_1}] \lw \  \cup \  \big\{\acce(\ell) \geq +\infty \big\}\ \text{with}\  \phi \in \{\sin, \cos, \tan, \log, \ldots\} \quad (\textsc{Math}^\prime)
$$
 $$\begin{array}{c}
\mathcal{C}^\prime\left[ x \texttt{:=}^\ell e^{\ell_1}\right]  \lw= \big(C, \lw \left[x \mapsto \ell\right] \big) \
 \text{where}\ C =  \mathcal{E}^\prime[e_1^{\ell_1}] \lw  \cup \left\{\acce(\ell_1) \geq \acce(\ell) \right\}
\end{array} \quad (\textsc{Assign}^\prime)
$$
 $$
\mathcal{C}^\prime\left[c_1^{\ell_1} \texttt{;} c_2^{\ell_2}  \right]  \lw= \big(C_1 \cup C_2 , \lw_2\big)
\ \text{with}\ \big(C_1, \lw_1 \big) = \mathcal{C}^\prime\left[c_1^{\ell_1}\right]\lw \ \text{and}\ \big(C_2, \lw_2 \big) = \mathcal{C}^\prime\left[c_2^{\ell_2}\right]\lw_1
\quad (\textsc{Seq}^\prime)
$$
\begin{equation*}
\begin{array}{c}
\mathcal{C}^\prime[\texttt{if}^\ell\ e^{\ell_0}\ \texttt{then}\ c^{\ell_1}\ \texttt{else}\ c^{\ell_2}]\ \lw =
 (C_1\cup C_2\cup C_3,\lw')\ \\ \text{where}\ \left|
\begin{array}{l}
 \forall x\in\text{Id},\ \lw'(x)=\ell,\
(C_1,\lw_1)=\mathcal{C}^\prime[c_1^{\ell_1}]\ \lw,\
(C_2,\lw_2)=\mathcal{C}^\prime[c_2^{\ell_2}]\ \lw,\\
C_3= \underset{x\in \text{Id}}{\bigcup} \left\{
 \acce(\lw_1(x)) \geq \acce(\ell), \
  \acce(\lw_2(x)) \geq \acce(\ell)
\right\}
\end{array}\right.\end{array}\quad (\textsc{Cond}^\prime)
\end{equation*}
\begin{equation*}
\begin{array}{c}
\mathcal{C}^\prime[\texttt{while}^\ell\ e^{\ell_0}\ \texttt{do}\ c^{\ell_1}]\ \lw =
 (C_1\cup C_2,\lw')\ \\ \text{where}\ \left|
\begin{array}{l}
 \forall x\in\text{Id},\ \lw'(x)=\ell, \
(C_1,\lw_1)=\mathcal{C}^\prime[c_1^{\ell_1}]\ \lw'\\
C_2= \underset{x\in \text{Id}}{\bigcup} \left\{
 \acce(\lw(x)) \geq \acce(\ell),\
  \acce(\lw_1(x)) \geq \acce(\ell)
\right\}
\end{array}\right.\end{array}\quad (\textsc{While}^\prime)
\end{equation*}
$$
\mathcal{C}^\prime[\texttt{require\_nsb}(x, \texttt{p})^\ell ] \lw = \emptyset \quad (\textsc{Req}^\prime)
$$
\center\noindent\rule{10cm}{0.2mm}\\
\vspace{0.2cm}
$\xi(\ell)(\ell_1, \ell_2) = \min\left(\begin{array}{l}\max\big(\ufp(\ell_2) - \ufp(\ell_1) + \accd(\ell_1) - \accd(\ell_2) - \acce(\ell_2), 0\big),\\
                                                         \max\big(\ufp(\ell_1) - \ufp(\ell_2) + \accd(\ell_2) - \accd(\ell_1) - \acce(\ell_1), 0\big),1
\end{array}\right)$
\vspace{0.3cm}
\HRule
\caption{\label{ip}Constraints solved by PI with $\min$ and $\max$ carry bit formulation.}
\end{figure}
\subsection{Policy Iteration for Optimized Carry Bit Propagation}\label{sec42}

The policy iterations algorithm is used to solve nonlinear fixpoint equations when the function is written as the infimum of functions for which a fixpoint can be easily computed. The infimum formulation makes the function not being differentiable in the classical sense. The one proposed in~\cite{CGGMP05} to solve smallest fixpoint equations in static analysis requires the fact that the function is order-preserving to ensure the decrease of the intermediate solutions provided by the algorithm. In this article, because of the nature of the semantics, we propose a policy iterations algorithm for a non order-preserving function.

More precisely, let $F$ be a map from a complete $L$ to itself such that $F=\inf_{\pi\in \Pi} f^\pi$. Classical policy iterations solve $F(\x)=\x$ by generating a sequence $(\x^k)_k$ such that $f^{\pi^k}(\x^k)=\x^k$ and $\x^{k+1}<\x^k$. The set $\Pi$ is called the set of policies and $f^\pi$ a policy map (associated to $\pi$). The set of policy maps has to satisfy the selection property meaning that for all $\x\in L$, there exists $\pi \in\Pi$ such that $F(\x)=f^\pi(\x)$. This is exactly the same as for each $\x\in L$, the minimization problem $\Min_{\pi\in\Pi} f^\pi(\x)$ has an optimal solution. If $\Pi$ is finite and $F$ is order-preserving, policy iterations converge in finite time to a fixpoint of $F$. The number of iterations is bounded from above by the number of policies. Indeed, a policy cannot be selected twice in the running of the algorithm. This is implied by the fact that the smallest fixpoint of a policy map is computed.
In this article, we adapt policy iterations to the problem of precision tuning.
The function $F$ here is constructed from inequalities depicted in Figure~\ref{ilp} and Figure~\ref{ip}. We thus have naturally constraints of the form $F(\x)\leq \x$. We will give details about the construction of $F$ at Proposition~\ref{prop-ipfun}. Consequently, we are interested in solving:
\begin{equation}\small
\label{pbfond}
\begin{array}{lll}
\displaystyle{\Min_{\accd,\acce}} & & \displaystyle{\sum_{\ell} \accd(\ell)}\\
&\st& F\left(\begin{array}{c}\accd\\ \acce\end{array}\right)\leq \left(\begin{array}{c}\accd\\ \acce\end{array}\right)\quad
\accd\in\nn^{Lab},\ \acce\in\nn^{Lab}
\end{array}
\end{equation}

Let $\xi:Lab\to \{0,1\}$. We will write $S_\xi^1$ the system of inequalities depicted in Figure~\ref{ilp} and $S_\xi^2$ the system of inequalities presented at Figure~\ref{ip}. Note that the final system of inequalities is $S_\xi=S_\xi^1\cup S_\xi^2$ meaning that we add new constraints to $S_\xi^1$. If the system $S_\xi^1$ is used alone, $\xi$ is the constant function equal to 1. Otherwise, $\xi$ is defined by the formula at the end of Figure~\ref{ip}.
\begin{proposition}
\label{prop-ipfun}
The following results hold:
\begin{enumerate}
\item Let $\xi$ the constant function equal to 1. The system $S_\xi^1$ can be rewritten as $\{\accd\in\nn^{Lab}\mid F(\accd)\leq (\accd)\}$ where $F$ maps $\rr^{Lab}\times \rr^{Lab}$ to itself, $F(\nn^{Lab}\times \nn^{Lab})\subseteq (\nn^{Lab}\times \nn^{Lab})$ and has coordinates which are the maximum of a finite family of affine order-preserving functions.
\item Let $\xi$ the function such that $\xi(\ell)$ equals the function defined at Figure~\ref{ip}. The system $S_\xi$ can be rewritten as $\{(\accd,\acce)\in\nn^{Lab}\times \nn^{Lab}\mid F(\accd,\acce)\leq (\accd,\acce) \}$ where $F$ maps $\rr^{Lab}\times \rr^{Lab}$ to itself, $F(\nn^{Lab}\times \nn^{Lab})\subseteq (\nn^{Lab}\times \nn^{Lab})$ and all its coordinates are the min-max of a finite family of affine functions.
\end{enumerate}
\end{proposition}


From Proposition~\ref{prop-ipfun}, when $S_\xi$ is used, we can write $F$ as $F=\min_{\pi\in\Pi} f^\pi$, where $f^\pi$ is the maximum of a finite family of affine functions and thus used a modified policy iterations algorithm. The set of policies here is a map $\pi:Lab\mapsto \{0,1\}$. A choice is thus a vector of 0 or 1. A policy map $f^\pi$ is a function $\nn^{Lab}$ to itself such that the coordinates are $f_{\ell}^\pi(\ell)$. If the coordinate $f_{\ell}^\pi(\ell)$ depends on $\xi$ then $\xi(\ell)=\pi(\ell)$. Otherwise, the function is the maximum of affine functions and a choice is not required.

\begin{algorithm}\scriptsize
\SetAlgoLined
\KwResult{An over-approximation of an optimal solution of Equation~\eqref{pbfond}}
 Let $k:=0$, $S:=+\infty$\;
 Choose $\pi^0\in \Pi$\;
 Select an optimal solution of $(\accd^k,\acce^k)$ the integer linear program:
  \[
 \Min \left\{\sum_{\ell\in Lab} \accd(\ell)\mid f^{\pi^k}(\accd,\acce)\leq (\accd,\acce),\ \accd\in\nn^{Lab},\ \acce\in\nn^{Lab}\right\}\enspace ;
 \]
 \eIf{$\sum_{\ell\in Lab} \accd^k(\ell)<S$}{
 $S:=\sum_{\ell\in Lab} \accd^k(\ell)$\;
 Choose $\pi^{k+1}\in \Pi$ such that $F(\accd^k,\acce^k)=f^{\pi^{k+1}}(\accd^k,\acce^k)$\;
 $k:=k+1$ and go to 3\;
 }
 {Return $S$ and $\accd^k$.}
 \caption{Non-monotone Policy Iterations Algorithm}\label{PIalgo}
\end{algorithm}

\begin{proposition}[Algorithm correctness]
The sequence $(\sum_{\ell\in Lab} \accd^k(\ell))_{0\leq k\leq K}$ generated by Algorithm~\ref{PIalgo} is of finite length (i.e. $K\in\nn$) and satisfies a strict decrease before convergence:
$\sum_{\ell\in Lab} \accd^{k+1}(\ell)<\sum_{\ell\in Lab} \accd^k(\ell)$ if $k<K-1$ and $\sum_{\ell\in Lab} \accd^{K}(\ell)=\sum_{\ell\in Lab} \accd^{K-1}(\ell)$. The number of terms is smaller than the number of policies.
\end{proposition}

\begin{proof}\small
Let $\sum_{\ell\in Lab} \accd^k(\ell)$ be a term of the sequence and $(\accd^k,\acce^k)$ be the optimal solution of $\Min\{\sum_{\ell\in Lab} \accd(\ell)\mid f^{\pi^k}(\accd,\acce)\leq (\accd,\acce),\ \accd\in\nn^{Lab},\ \acce\in\nn^{Lab}\}$. Then $F(\accd^k,\acce^k)\leq f^{\pi^k}(\accd^k,\acce^k)$ by definition of $F$. Moreover,\\ $F(\accd^k,\acce^k)=f^{\pi^{k+1}}(\accd^k,\acce^k)$ and  $f^{\pi^k}(\accd^k,\acce^k)\leq (\accd^k,\acce^k)$. It follows that $f^{\pi^{k+1}}(\accd^k,\acce^k)\leq (\accd^k,\acce^k)$ and $(\accd^k,\acce^k)$ is feasible for the minimisation problem for which $(\accd^{k+1},\acce^{k+1})$ is an optimal solution. We conclude that $\sum_{\ell\in Lab} \accd^{k+1}(\ell)\leq \sum_{\ell\in Lab} \accd^k(\ell)$ and the Algorithm terminates if the equality holds or continues as the criterion strictly decreases. Finally, from the strict decrease, a policy cannot be selected twice without terminating the algorithm. In conclusion, the number of iterations is smaller than the number of policies.\hfill$\Box$
\end{proof}

Figure ~\ref{ip} displays the new rules that we add to the global system of constraints in the case where we optimize the carry bit of the elementary operations. Before introducing the optimized function $\xi$, the definition of the unit in the last place $\ulp$ of a real number $x$ is defined in Equation(\ref{ulp}).
\begin{equation}\label{ulp}\small
    \ulp(x) = \ufp(x)- \accd(x)+1 \enspace.
\end{equation}

Indeed, during an operation between two numbers, the $\xi$ function is computed as follows: If the $\ulp$ of one of the two operands is greater than the $\ufp$ of the other one (or conversely) then the two numbers are not aligned and no carry bit can be propagated through the operation (otherwise $\xi = 1$). This idea is presented in Figure \ref{ip} in which $\xi$ is formulated by $\min$ and $\max$ operators. Winning one bit may seen a ridiculous optimization at first sight. However, when many operations are done in a program which has to compute with some tens of $\accd$, this is far from being negligible. Formally, let $x$ and $y$ be the operands of some operation whose result is $z$. The errors are $\varepsilon(x)$, $\varepsilon(y)$ and $\varepsilon(z)$ respectively on $x$, $y$ and $z$.
Note that the behaviour of $\xi$ is also valid in the case of propagation of the carry bit on the errors. Consequently, using the definitions of $\ufpe$ and $\acce$ introduced in Section \ref{sec31}, we have $\ufpe(x) = \ufp(x) - \accd(x)$ and $\ulpe(x) = \ufpe(x) - \acce(x) + 1$ with $\ufpe$ and $\ulpe$ are the unit in the first and in the last place of the errors, respectively, for $x$ (same reasoning for $y$ and $z$). The optimized function $\xi$ of Figure~\ref{ip} derives from Equation~(\ref{xi}).
\begin{equation}\label{xi}\small
	\xi (z, x, y) = \begin{cases} 0 &\ufpe(x) - \acce(x) \ge \ufp(y) - \accd(y) ~~ or~~ \text{conversely}, \\
		1 & \text{otherwise.} \end{cases}
\end{equation}
As mentioned in Equation~(\ref{xi}), to compute the $\ulp$ of the errors on the operands, we need to estimate the number of bits of the error $\acce$ for each operand which explains the new rules and constraints given in Figure \ref{ip}. Hence, these constraints are complementary to the rules of Figure \ref{ilp}, already explained in Section~\ref{sec41}, in which the only difference is that we activate the new function $\xi$ instead of its over-approximation of Figure~\ref{ilp}. Let us concentrate on the rules of Figure~\ref{ip}. The function $\mathcal{E}^\prime[e]\ \lw$ generates the new set of constraints for an expression $e \in Expr$ in the environment $\lw$. For Rule $(\textsc{Const}^\prime)$, the number of significant bits of the error $\acce = 0$ whereas we impose that the $\acce$ of a variable $x$ at control point $\ell$ is less than the last assignment of $\acce$ in $\lw(x)$ as shown in Rule $(\textsc{Id}^\prime)$ of Figure~\ref{ip}. Considering Rule $(\textsc{Add}^\prime)$, we start by generating the new set of constraints $\mathcal{E}^\prime[e_1^{\ell_1}] \lw$ and $\mathcal{E}^\prime[e_2^{\ell_2}] \lw$ on the operands at control points $\ell_1$ and $\ell_2$. Then, we require that $\acce(\ell) \geq \acce(\ell_1)$ and $\acce(\ell) \geq \acce(\ell_2) $ where the result of the addition is stored at control point $\ell$. Additionally, the number of significant bits of the error on the result $\acce(\ell)$ at control point $\ell$ is computed as shown hereafter.
\begin{equation*}\small
  \acce(\ell) \geq \max\left(
  \begin{array}{l}
\ufp(\ell_1) - \accd(\ell_1)\\
 \ufp(\ell_2) - \accd(\ell_2)
 \end{array}\right)
  - \min\left(
  \begin{array}{l}\ufp(\ell_1) - \accd(\ell_1) - \acce(\ell_1) \\ \ufp(\ell_2) - \accd(\ell_2) - \acce(\ell_2)
  \end{array}\right)
   + \xi(\ell)(\ell_1, \ell_2)
\end{equation*}
By breaking the $\min$ and $\max$ operators, we obtain the constraints on $\acce(\ell)$ of Rule $(\textsc{Add}^\prime)$. For the subtraction, the constraints generated are similar to the addition case. Considering now Rule $(\textsc{Mult}^\prime)$, as we have defined in  Section \ref{sec41}, $\varepsilon(z) = x \cdot \varepsilon(y) + y \cdot \varepsilon(x) + \varepsilon(x) \cdot \varepsilon(y)$ where $z$ is the result of the product of $x$ and $y$. By reasoning on the $\ulp$ of the error, we bound $\varepsilon(z)$ by
\begin{equation*}\small
\begin{aligned}
\varepsilon(z) = 2^{\ufp(x)} \cdot 2^{\ufp(y) - \accd(y)- \acce(y) + 1} + 2^{\ufp(y)}\cdot 2^{\ufp(x) - \accd(x) - \acce(x) + 1} \\+ 2^{\ufp(y) + \ufp(x) - \accd(x) -\accd(y) - \acce(x) - \acce(y)+2}
\end{aligned}
\end{equation*}
By selecting the smallest term $\ufp(y) + \ufp(x) - \accd(x) -\accd(y) - \acce(x) - \acce(y)+2$, we obtain that
\begin{equation*}\small
  \acce(\ell) \geq \max\left(
  \begin{array}{l}
\ufp(\ell_1) + \ufp(\ell_2) - \accd(\ell_1)\\
 \ufp(\ell_1) + \ufp(\ell_2) - \accd(\ell_2)
 \end{array}\right)
  -
  \begin{array}{l}\ufp(\ell_1) + \ufp(\ell_2) - \accd(\ell_1)- \\ \accd(\ell_2) - \acce(\ell_1) - \acce(\ell_2)\\ + 2
  \end{array} \enspace.
\end{equation*}
Finally, by simplifying the equation above we found the constraints of Rule $(\textsc{Mult}^\prime)$ in Figure \ref{ip} (same for Rule $(\textsc{Div}^\prime)$).
\normalsize
For Rule $(\textsc{Sqrt}^\prime)$, we generate the constraints on the expression $\mathcal{E}^\prime[e_1^{\ell_1}] \lw$ and we require that $\acce$ of the result stored at control point $\ell$ is greater than the $\acce$ of the expression a control point $\ell_1$. For Rule $(\textsc{Math}^\prime)$ , we assume that $\acce(\ell)$ is unbounded. Concerning the commands, we define the set $\mathcal{C}^\prime[c]\ \lw$ which has the same function as $\mathcal{C}$ defined in Figure \ref{ilp}. The reasoning on the commands also remains similar except that this time we reason on the number of bits of the errors $\acce$. The only difference is in Rule $(\textsc{Req}^\prime)$ where the set of constraints is empty. Let us recall that the constraints of Figure \ref{ip} are added to the former constraints of Figure \ref{ilp} and are sent to a linear solver (GLPK in practice). 
  \begin{figure}[tb]
\small
\noindent\rule{12.4cm}{0.25mm}
\begin{equation*}
  \frac{\lw(x) = c\texttt{\#p}}{\langle x^{\ell}, \lw \rangle \longrightarrow \langle c^{\ell}\texttt{\#p}, \lw \rangle}
\end{equation*}
\begin{equation*}
\frac{ c =  c_1 \odot c_2,\ \texttt{p} =  \ufp(c) - \ufpe\big(c^\ell\texttt{\#p}\big)}
{\langle c_1^{\ell_1}\texttt{\#$p_1$} \odot^\ell c_2^{\ell_2}\texttt{\#$p_2$}, \lw \rangle \longrightarrow \langle c\texttt{\#p}, \lw \rangle} \quad \odot \in \{+, -, \times, \div\}
\end{equation*}
\begin{equation*}
\frac{\langle e_1^{\ell_1}, \lw \rangle \longrightarrow \langle e_1^{\prime\ell_1}, \lw \rangle}
     {\langle e_1^{\ell_1} \odot^\ell e_2^{\ell_2}, \lw \rangle   \longrightarrow   \langle e_1^{\prime\ell_1} \odot^\ell e_2^{\ell_2}, \lw \rangle} \qquad
\frac{\langle e_2^{\ell_2}, \lw \rangle \longrightarrow \langle e_2^{\prime\ell_2}, \lw \rangle}
     {\langle c_1^{\ell_1}\texttt{\#p} \odot^\ell e_2^{\ell_2}, \lw \rangle \longrightarrow \langle c_1^{\ell_1}\sharp\texttt{p} \odot^\ell e_2^{\prime\ell_2}, \lw \rangle}
\end{equation*}
\begin{equation*}
  \frac{\langle e^{\ell_1}, \lw \rangle \longrightarrow \langle e^{\prime\ell_1}, \lw \rangle}{\langle\phi(e^{\ell_1})^\ell, \lw\rangle \longrightarrow \langle\phi(e^{\prime\ell_1})^\ell, \lw\rangle }
\qquad \frac{ c = \phi(c_1) \quad \texttt{q} = \ufp(c) - \ufpe(c^\ell\texttt{\#p}) }{\langle \phi(c_1^{\ell_1}\texttt{\#p})^\ell, \lw \rangle \longrightarrow \langle c^\ell\texttt{\#q}, \lw \rangle}\quad \phi \in \{\sin, \cos,\ldots\}
\end{equation*}

\begin{equation*}
  \frac{\langle e^{\ell_1}, \lw \rangle \longrightarrow \langle e^{\prime\ell_1}, \lw \rangle}{\langle\sqrt{e^{\ell_1}}^\ell, \lw\rangle \longrightarrow \langle\sqrt{e^{\prime\ell_1}}^\ell, \lw\rangle }
\qquad
\frac{c = \sqrt{c_1} \quad \texttt{q} = \ufp(c) - \ufpe(c^\ell\texttt{\#p}) }{\langle \sqrt{c_1^{\ell_1}\texttt{\#p}}^\ell, \lw \rangle \longrightarrow \langle c^\ell\texttt{\#q}, \lw \rangle}
\end{equation*}
\noindent\rule{12.4cm}{0.25mm}
 \caption{Small Step Operational semantics of arithmetic expressions.}\label{semantic}
\end{figure}

\section{Correctness} \label{sec5}

In this section, we present proofs of correctness concerning the soundness of the analysis (Section~\ref{sec51}) and the integer nature of the solutions (Section \ref{sec52}).

\subsection{Soundness of the Constraint System}\label{sec51}

Let $\equiv$ denote the syntactic equivalence and let $e^\ell \in Expr$ be an expression. We write $\const(e^\ell)$ the set of constants occurring in the expression $e^\ell$. For example, $\const(18.0^{\ell_1}\times^{\ell_2} x^{\ell_3} +^{\ell_4} 12.0^{\ell_5}\times^{\ell_6}y^{\ell_7} +^{\ell_8} z^{\ell_9}) =\{18.0^{\ell_1}, 12.0^{\ell_5}\}$. Also, we denote by $\tau\ :\ \text{Lab}\rightarrow \mathbb{N}$ a function mapping the labels of an expression to a $\accd$. The notation $\tau \models \mathcal{E}[ e^\ell] \lw$ means that $\tau$ is the minimal solution to the ILP
$\mathcal{E}[ e^\ell] \lw$. We write $\lw_\perp$ the empty environment ($dom(\lw_\perp) = \emptyset$).

The small step operational semantics of our language is displayed in Figure~\ref{semantic}. It is standard, the only originality being to indicate explicitly the $\accd$ of constants. For the result of an elementary operation, this $\accd$ is computed in function of the $\accd$ of the operands. Lemma \ref{lem1} below asses the soundness of the constraints for one step of the semantics.
\begin{lem}\label{lem1}
Given an expression $e^\ell \in Expr$, if $e^\ell \rightarrow e^{\prime\ell}$ and $\tau \models \mathcal{E}[ e^\ell] \lw_\perp$  then
for all $c^{\ell_c}$\texttt{\#p} $\in$ $\const(e^{\prime\ell})$ we have \texttt{p} = $\tau(\ell_c)$.
\end{lem}
\begin{proof}\small
By case examination of the rules of Figure \ref{ilp}. Hereafter, we focus on the most interesting case of addition of two constants. Recall that $\ufpe(\ell) = \ufp(\ell) - \accd(\ell)$ for any control point $\ell$.
Assuming that $e^\ell \equiv c_1^{\ell_1} +^\ell c_2^{\ell_2}$ then by following the reduction rule of Figure \ref{semantic}, we have $e^\ell \rightarrow c^{\ell}\texttt{\#p}$ with $\texttt{p} =  \ufp(c) - \ufpe\big(c\big)$.
On the other side, by following the set of constraints of Rule (\textsc{Add}) in Figure \ref{ilp} we have
$\small\mathcal{E}[e^\ell] \lw =
 \{\accd(\ell_1) \geq \accd(\ell) + \ufp(\ell_1) - \ufp(\ell) +  \xi(\ell)(\ell_1, \ell_2),
\hspace{0.1cm}\accd(\ell_2) \geq \accd(\ell) + \ufp(\ell_2)
$ $ - \ufp(\ell) + \xi(\ell)(\ell_1, \ell_2)
 \}$.
 These constraints can be written as
 \begin{eqnarray*}\small
 \accd(\ell) &\leq & \ufp(\ell) - \ufp(\ell_1) + \accd(\ell_1) - \xi(\ell)(\ell_1, \ell_2) \\
 \accd(\ell) &\leq & \ufp(\ell) - \ufp(\ell_2) + \accd(\ell_2) - \xi(\ell)(\ell_1, \ell_2)
 \end{eqnarray*}
 and may themselves be rewritten as Equation~(\ref{accd}), i.e. $$\small \accd(\ell) \leq \ufp(\ell) - \max\big(\ufp(\ell_1) - \accd(\ell_1), \ufp(\ell_2) - \accd(\ell_2) \big) - \xi(\ell)(\ell_1, \ell_2)\enspace .$$
 Since, obviously, $\ufp(c) = \ufp(\ell)$ and since the solver finds the minimal solution to the ILP, it remains to show that
 \[\small\ufpe(\ell) = \max\big(\ufp(\ell_1) -\accd(\ell_1), \ufp(\ell_2) -\accd(\ell_2), \ufp(\ell)-\prec(\ell)\big) -\xi(\ell)(\ell_1, \ell_2)\]
which corresponds to the assertion of Equation(\ref{tunisienneprime}). Consequently, $\accd(\ell) = \texttt{p}$ as required, for this case, in Figure \ref{semantic}.\hfill$\Box$
\end{proof}
\begin{theorem}\label{theorem1}
Given an expression $e^\ell \rightarrow e^{\prime\ell}$. If $e^\ell \rightarrow^* e^{\prime\ell}$ and if  $\tau \models \mathcal{E}[ e^\ell] \lw_\perp$ , then $\forall$ $c^{\ell_c}\texttt{\#p}$ $\in$ $\const(e^{\prime\ell})$ we have $\texttt{p} = \tau(\ell_c)$.
\end{theorem}
\begin{proof}\small
By recurrence on the length of the reduction path. \hfill$\Box$
\end{proof}
\subsection{ILP Nature of the Problem}\label{sec52}
In this section, we give insights about the complexity of the problem. The computation relies on integer linear programming. Integer linear programming is known to belong to the class of NP-Hard problems. A lower bound of the optimal value in a minimization problem can be furnished by the continuous linear programming relaxation. This relaxation is obtained by removing the integrity constraint. Recall that a (classical) linear program can be solved in polynomial-time. Then, we can solve our problem in polynomial-time if we can show that the continuous linear programming relaxation of our ILP has an unique optimal solution with integral coordinates.
Proposition~\ref{propLP} presents a situation where a linear program has a unique optimal solution which is a vector of integers.

\begin{proposition}
\label{propLP}
Let $G:[0,+\infty)^d\mapsto [0,+\infty)^d$ be an order-preserving function such that $G(\nn^d)\subseteq \nn^d$. Suppose that the set $\{y\in\nn^d\mid G(y)\leq y\}$ is non-empty. Let $\varphi:\rr^d\mapsto \rr$ a strictly monotone function such that $\varphi(\nn^d)\subseteq \nn$. Then, the minimization problem:
\[
\displaystyle{\Min_{y\in [0,+\infty)^d}}\ \varphi(y)\ \st G(y)\leq y
\]
has an unique optimal solution which is integral.
\end{proposition}

\begin{proof}\small
Let $L:=\{x\in [0,+\infty)^d\mid G(x)\leq x\}$ and $u=\inf L$. It suffices to prove that $u\in\nn^d$. Indeed, as $\varphi$ is stricly monotone then $\varphi(u)<\varphi(x)$ for all $x\in [0,+\infty)^d$ s.t. $G(x)\leq x$ and $x\neq u$. The optimal solution is thus $u$.
If $u=0$, the result holds. Now suppose that $0<u$, then $0\leq G(0)$. Let $M:=\{y\in \nn^d\mid y\leq G(y), y\leq u\}$. Then $0\in M$ and we write $v:=\sup M$. As $M$ is a complete lattice s.t. $G(M)\subseteq M$, from Tarski's theorem, $v$ satisfies $G(v)=v$ and $v\leq u$. Moreover, $v\in\nn^d$ and $v\leq u$. Again, from Tarski's theorem, $u$ is the smallest fixpoint of $G$ then it coincides with $v$. We conclude that $u\in\nn^d$.\hfill$\Box$
\end{proof}

\begin{theorem}
Assume that the system of inequalities depicted in Figure~\ref{ilp} has a solution. The smallest amount of memory $\sum_{\ell\in Lab} \accd(\ell)$ for the system of inequalities depicted in Figure~\ref{ilp} can be computed in polynomial-time by linear programming.
\end{theorem}
\begin{proof}\small
The function $\sum_{\ell\in Lab} \accd(\ell)$ is strictly monotone and stable on integers. From the first statement of Proposition~\ref{prop-ipfun}, the system of constraints is of the form $F(\accd)\leq \accd$ where $F$ is order-preserving and stable on integers. By assumption, there exists an vector of integers $\accd$ s.t. $F(\accd)\leq \accd$. We conclude from Proposition~\ref{propLP}. \hfill$\Box$
\end{proof}
For the second system, in practice, we get integral solutions to the continuous linear programming relaxation of our ILP of Equation~\eqref{pbfond}. However, because of the lack of monotonicity of the functions for the rules (ADD) and (SUB), we cannot exploit Proposition~\ref{propLP} to prove the polynomial-time solvability. 
\section{Experimental Results}\label{sec6}
\begin{table}[tb]
\centerline{
\scriptsize
\begin{tabular}{lrrrrrrrrrrrrrr}
\hline
Program & TH   & &BL & IEEE & ILP-time && BL & IEEE& PI-time & H & S & D & LD\\
\hline
          & $10^{-4} $ &\hspace{0.6cm}& 61\% & 43\% & 0.9s & \hspace{0.6cm}& 62\% & 45\% & 1.5s & 8 & 88 & 25 & 0 \\
          & $10^{-6} $ && 50\% & 21\% & 0.9s && 51\% & 21\% & 1.4s & 2 & 45 & 74 & 0 \\
\textbf{arclength} & $10^{-8} $ && 37\% & 3\%  & 0.8s && 38\% & 4\%  & 1.6s & 2 & 6  & 113& 0 \\
          & $10^{-10}$ && 24\% & -1\% & 1.0s && 25\% & -1\% & 1.7s & 2 & 0  & 116& 3 \\
          & $10^{-12}$ && 12\% & -17\%& 0.3s && 14\% & -8\% & 1.5s & 2 & 0  & 109& 10\\
\hline
        & $10^{-4}$  && 64\% & 45\% & 0.1s && 67\% & 56\% & 0.5s & 6  & 42 & 1  & 0 \\
        & $10^{-6}$  && 53\% & 30\% & 0.2s && 56\% & 31\% & 0.5s & 1  & 27 & 21 & 0 \\
\textbf{simpson}& $10^{-8}$  && 40\% & 4\%  & 0.1s && 43\% & 7\%  & 0.3s & 1  & 5  & 43 & 0 \\
        & $10^{-10}$ && 27\% & 1\%  & 0.1s && 28\% & 1\%  & 0.4s & 1  & 0  & 48 & 0 \\
        & $10^{-12}$ && 16\% & 1\%  & 0.1s && 16\% & 1\%  & 0.3s & 0  & 1  & 48 & 0 \\
\hline
             & $10^{-4}$ && 73\% & 61\% & 0.2s && 76\% & 62\% & 1.0s & 53 & 69 & 0 & 0 \\
             & $10^{-6}$ && 62\% & 55\% & 0.2s && 65\% & 55\% & 1.0s & 2  & 102& 0 & 0 \\
\textbf{accelerometer}       & $10^{-8}$ && 49\% & 15\% & 0.2s && 52\% & 18\% & 1.0s & 2  & 33 & 69 & 0 \\
             & $10^{-10}$&& 36\% & 1\%  & 0.2s && 39\% & 1\%  & 1.0s & 2  & 0  & 102& 0 \\
             & $10^{-12}$&& 25\% & 1\%  & 0.2s && 28\% & 1\%  & 1.0s & 2  & 0  & 102& 0\\
\hline
             & $10^{-4}$ && 78\% & 66\% & 0.08s && 79\% & 68\% & 1.3s & 46 & 38 & 0 & 0 \\
             & $10^{-6}$ && 67\% & 53\% & 0.08s && 68\% & 56\% & 0.5s & 12  & 70& 2 & 0 \\
\textbf{rotation}  & $10^{-8}$ && 53\% & 29\% & 0.07s && 54\% & 29\% & 0.4s & 0  & 46 & 38 & 0 \\
             & $10^{-10}$&& 40\% & 0\%  & 0.1s && 41\% & 0\%  & 0.5s & 0  & 0  & 84& 0 \\
             & $10^{-12}$&& 29\% & 0\%  & 0.09s && 30\% & 0\%  & 0.5s & 0  & 0  & 48& 0\\
\hline
                & $10^{-4}$ && 68\% & 46\% & 1.8s && 69\% & 46\% & 10.7s & 260 & 581 & 0 & 0 \\
                & $10^{-6}$ && 57\% & 38\% & 1.8s && 58\% & 45\% & 11.0s & 258  & 580& 3 & 0 \\
\textbf{lowPassFilter}  & $10^{-8}$ && 44\% & -7\% & 2.0s && 45\% & -7\% & 11.4s & 258  & 2 & 581 & 0 \\
                & $10^{-10}$&& 31\% & -7\%  & 1.7s && 32\% & -7\%  & 10.9s & 258  & 0  & 583& 0 \\
                & $10^{-12}$&& 20\% & -7\%  & 1.8s && 21\% & -7\%  & 11.3s & 258  & 0  & 583& 0\\
\hline
                & $10^{-4}$ && 71\% & 54\% & 0.15s && 71\% & 54\% & 0.4s & 0 & 13 & 0 & 0 \\
                & $10^{-6}$ && 60\% & 50\% & 0.2s && 60\% & 50\% & 0.5s & 0  & 12& 1 & 0 \\
\textbf{Pendulum}  & $10^{-8}$ && 47\% & 0\% & 012s && 47\% & 0\% & 0.4s &0  & 0 & 13 & 0 \\
                & $10^{-10}$&& 33\% & 0\%  & 0.16s && 34\% & 0\%  & 0.5 & 0  & 0  & 13& 0 \\
                & $10^{-12}$&& 22\% & 0\%  & 0.11s && 22\% & 0\%  & 0.4s & 0  & 0  & 13& 0\\
\hline
\end{tabular}
}   \vspace{0.3cm}\caption{\label{tab1}Precision tuning results for XXX for the ILP and PI methods.}
\end{table}
In this section, we aim at evaluating the performance of our tool XXX implementing the techniques of Section~\ref{sec4}. Recall that XXX reduces the problem of precision tuning to an ILP by generating the constraints defined in Section \ref{sec41} which can be solved by a linear solver. We use GLPK\footnote{https://www.gnu.org/software/glpk/} in practice with continuous variables since we know from Section~\ref{sec42} that it is possible. Indeed, GLPK's constant \texttt{col-kind} is set to \texttt{cv} and not to \texttt{iv} as for integer problems. Alternatively, a second optimized set of constraints may be generated by applying the PI technique introduced in Section \ref{sec42}.

We have evaluated XXX on several numerical programs. Two of them were used as a benchmark for precision tuning in prior work \cite{RGNNDKSBIH13} and are coming from the GNU scientific library (GSL): \textbf{arclength} which is a program first introduced in \cite{arc} and the \textbf{simpson} program which corresponds to an implementation of the widely used Simpson{'s} rule for integration in numerical analysis \cite{simpsons}. In addition, we have experimented the techniques introduced in this article on three programs used as benchmarks for XXX in its former version \cite{KM19,KM20,KMA19}. The program \textbf{rotation} performs a matrix-vector product to rotate a vector around the $z$ axis by an angle of $\theta$ \cite{KMA19}. The second program \textbf{accelerometer} is an application to measure an inclination angle with an accelerometer \cite{KM19} and, finally, \textbf{lowPassFilter} program \cite{KM20} which is taken from a pedometer application \cite{pedometer} that counts the number of footsteps of the user. The last two programs come from the IoT field. Also, we take the \textbf{pendulum} program already introduced in Section \ref{sec2}.

 The experiments shown in Table \ref{tab1} present the results of precision tuning returned by XXX for each error threshold $10^{-4}$, $10^{-6}$, $10^{-8}$ and $10^{-10}$. The error threshold represents the $\accd$. For instance, the result is required to be correct up to 4 digits for an error threshold of $10^{-4}$. Obviously, XXX counting the number of significant bits, these threshold are converted in binary. In Table \ref{tab1}, we represent by "TH" the error threshold given by the user. "BL" is the percentage of optimization at bit level. "IEEE" denotes the percentage of optimized variables in IEEE754 formats (\texttt{binary16}, \texttt{binary32}, $\ldots$ The $\accd$ obtained at bit level is approximated by the immediately the upper number of bits corresponding to a IEEE754 format). "ILP-time" is the total time of XXX analysis in the case of ILP formulation. We have also "PI-time" to represent the time passed by XXX to find the right policy and to resolve the precision tuning problem. "H", "S", "D" and "LD" denote respectively the number of variables obtained in, half, float, double and long-double precision when using the PI formulation that clearly displays better results.

\begin{table}[tb]
\scriptsize
\begin{center}
 \begin{tabular}{|c||l|c|c|c|c|c|c|c|c|c|}
  \hline \hline
  \multirow{1}{1.7cm}{Program}&\multirow{1}{2cm}{\centerline{Tool}}  &\multicolumn{4}{c|}{\#Bits saved - Time in seconds}\\
  \cline{3-6}
  && Threshold $10^{-4}$ &Threshold $10^{-6}$ &Threshold $10^{-8}$ & Threshold $10^{-10}$ \\
 \hline
 \multicolumn{1}{|l||}{\textbf{arclength}}
  &XXX ILP (28) & \textbf{2464b.} - 1.8s.& \textbf{2144b.} - 1.5s.& \textbf{1792b.} - 1.7s. & \textbf{1728b.} - 1.8s. \\
  &XXX SMT (22)& 1488b. - 4.7s. & 1472b. - 3.04s. & 864b. - 3.09s. & 384b. - 2.9s. \\
  &Precimonious (9)& 576b. - 146.4s. & 576b. - 156.0s. & 576b. - 145.8s. & 576b. - 215.0s. \\
 \hline
  \multicolumn{1}{|l||}{\textbf{simpson}}
  &XXX ILP (14)& \textbf{1344b.} - 0.4s.& \textbf{1152b.} - 0.5s. & \textbf{896b.} - 0.4s. & \textbf{896b.} - 0.4s. \\
  &XXX SMT (11)& 896b. - 2.9s. & 896b. - 1.9s. & 704b. - 1.7s. & 704b. - 1.8s.\\
  &Precimonious (10)& 704b. - 208.1s. & 704b. - 213.7s. & 704b. - 207.5s. & 704b. - 200.3s. \\
 \hline
 \multicolumn{1}{|l||}{\textbf{rotation}}
  &XXX ILP (25)& \textbf{2624b.} - 0.47s.& 2464b. - 0.47s. & 2048b. - 0.54s. & 1600b. - 0.48s. \\
  &XXX SMT (22)& 1584b. - 1.85s. & 2208b. - 1.7s. & 1776b. - 1.6s. & 1600b. - 1.7s.\\
  &Precimonious (27)& 2400b. - 9.53s. & \textbf{2592b.} - 12.2s. & \textbf{2464b.} - 10.7s. & \textbf{2464b.} - 7.4s. \\
  \hline
  \multicolumn{1}{|l||}{\textbf{accel.}}
  &XXX ILP (18)& \textbf{1776b.} - 1.05s. & \textbf{1728b.} - 1.05s. & \textbf{1248b.} - 1.04s. & \textbf{1152b.} - 1.03s. \\
  &XXX SMT (15)& 1488b. - 2.6s. & 1440b. - 2.6s. & 1056 - 2.4s. & 960b. - 2.4s. \\
  &Precimonious (0)& - & - & - &-  \\
  \hline
  \hline
  \end{tabular}
  \end{center}
  \caption{\label{tabcomp}Comparison between XXX ILP, XXX SMT and Precimonious: number of bits saved
  by the tool and time in seconds for analyzing the programs.}
\end{table}

Let us focus on the first "TH", "BL", "IEEE" and "ILP-time" columns of Table \ref{tab1}. These columns have been measured for the case of ILP by being pessimistic on the propagation of carry bit as described in Section \ref{sec41}. We recall that all variables of our programs are in double precision before analysis. For the \textbf{arclength} program, the final percentage of bits after optimization reaches $61\%$ at bit level while it achieves $43\%$ in IEEE formats ($100\%$ is the percentage of all variables initially in double precision, $121$ variables for the original \textbf{arclength} program that used $7744$ bits). This is obtained in only $0.9$ second by applying the ILP formulation. As another option, when we refine the solution by applying the policy iteration method (from the sixth column), XXX displays better results attaining $62\%$ at bit level and $43\%$ for the IEEE formats. Although XXX needs more time of analysis to find and iterate between policies, the time of analysis remain negligible, not exceeding $1.5$ seconds. For a total of $121$ variables for the \textbf{arclength} original program, XXX succeeds in tuning $8$ variables to half precision (H), the majority with $88$ variables passes to simple precision (S) whereas $25$ variables remain in double precision (D) for an error threshold of $10^{-4}$. We remark that our second method displays better results also for the other user error thresholds. For the \textbf{simpson}, \textbf{accelerometer}, \textbf{rotation} and \textbf{lowPassFilter}, the improvement is also more important when using the PI technique than when using the ILP formulation. For instance, for an error threshold of $10^{-6}$ for the \textbf{simpson} program, only one variable passes to half precision, $27$ variables turns to simple precision while $21$ remains in double precision with $56\%$ of percentage of total number of bits at bit level using the policy iteration method. Concerning the \textbf{pendulum} code, the two techniques return the same percentage at bit level and IEEE754 format for the majority of error thresholds (except $10^{-10}$ where XXX reaches $34\%$ at bit level when using the PI method).

Now, we stress on the negative percentage that we obtain in Table \ref{tab1}, especially for the \textbf{arclength} program with $10^{-10}$ and $10^{-12}$ for the columns IEEE and for the \textbf{lowPassFilter} program for errors of $10^{-8}$, $10^{-10}$ and $10^{-12}$. In fact,
XXX being able to return new formats for any precision required by the user without additional cost nor by increasing the
complexity even if it fails to have a significant improvement on the program output. To be specific, taking again the \textbf{arclength} program, for an error of $10^{-12}$, XXX fulfills this requirement by informing the user that this precision is achievable only if 10 variables passes to the long double precision (LD) which is more than the original program whose variables are all in double precision. By doing so, the percentage of IEEE formats for both ILP and PI formulations reaches $-17\%$ and $-8\%$, respectively. Same reasoning is adopted for the \textbf{lowPassFilter} which spends more time, nearly $12$ seconds, with the policy iteration technique to find the optimized formats (total of 841 variables). Note that in these cases, other tools like Precimonious \cite{RGNNDKSBIH13} fail to propose formats.

Table~\ref{tabcomp} shows a comparison between the new version of XXX combining both ILP and PI formulations, the former version of XXX that uses the Z3 SMT solver coupled to binary search to find optimal solution \cite{KMA19} and the prior state-of the-art Precimonious \cite{RGNNDKSBIH13}. The results of the mixed precision tuning are shown for the \textbf{arclength}, \textbf{simpson}, \textbf{rotation} and \textbf{accelerometer} programs. Since XXX and Precimonious implement two different techniques, we have adjusted the criteria of comparison in several points. First, we mention that XXX optimizes much more variables than Precimonious. While it disadvantages XXX, we only consider in the experiments of Table~\ref{tabcomp} the variables optimized by Precimonious to estimate the quality of the optimization. Second, let us note that the error thresholds are expressed
in base $2$ in XXX and in base $10$ in Precimonious. For the relevance of comparisons, all the thresholds
are expressed in base $10$ in tables \ref{tab1} and \ref{tabcomp}. In practice, XXX will use the base $2$ threshold immediately lower than the required base $10$ threshold.

In Table \ref{tabcomp}, we indicate in bold the tool that exhibits better results for each error threshold and each program. Starting with the \textbf{arclength} program, XXX ILP displays better results than the other tools by optimizing $28$ variables. For an error threshold of $10^{-4}$, $2464$ bits are saved by XXX ILP in $1.8$ seconds while XXX Z3 saved only 1488 bits in more time ($11$ seconds). Precimonious were the slowest tool on this example, more than $2$ minutes with $576$ bits for only $9$ variables optimized. For the \textbf{simpson} program, XXX ILP do also better than both other tools. However, for the \textbf{rotation} program, XXX ILP saves more bits than the other tools only for an error of $10^{-4}$  while Precimonious do well for this program for the rest of error thresholds. Finally, Precimonious fails to tune the \textbf{accelerometer} program (0 variables) at the time that XXX ILP do faster (only 1 second) to save much more bits than XXX SMT for the four given error thresholds. 
\section{Conclusion and Perspectives} \label{sec7}
In this article, we have introduced a new technique for precision tuning, totally different from
the existing ones. Instead of changing more or less randomly the data types of the numerical variables and
running the programs to see what happens, we propose a semantical modelling of the propagation of the
numerical errors throughout the code. This yields a system of constraints whose minimal solution
gives the best tuning of the program, furthermore, in polynomial time. Two variants of this system are proposed. The first
one corresponds to a pure ILP. The second one, which optimizes the propagation of carries in the elementary
operations can be solved using policy iterations \cite{CGGMP05}. Proofs of correctness concerning the soundness of the
analysis and the integer nature of the solutions have been presented in Section \ref{sec5} and experimental
results showing the efficiency of our method have been introduced in Section \ref{sec6}.

Compared to other approaches, the strength of our method is to find directly the minimal number of bits
needed at each control point to get a certain accuracy on the results. Consequently, it is not dependant of
a certain number of data types (e.g. the IEEE754 formats) and its complexity does not increase as the
number of data types increases.
The information provided may also be used to generate computations in the fixpoint arithmetic with
an accuracy guaranty on the results. Concerning scalability, we generate a linear number of constraints and variables in the size of the analyzed program. The only limitation is the size of the problem accepted by the solver.
Note that the number of variables could be reduced by assigning the same precision to a whole piece of code (for example an arithmetic expression, a line of code, a function, etc.) Code synthesis for the fixpoint arithmetic
and assigning the same precision to pieces of code are perspectives we aim at explore at short term.

At longer term, other developments of the present work are planned. First we wish to adapt the
techniques developed in this article to the special case of Deep Neural Networks for which it is
important to save memory usage and computational resources. Second, we aim at  using our
precision tuning method to guide lossy compression techniques for floating-point datasets \cite{DFHS19}. In this
case, the bit-level accuracy inferred by our method would determine the compression rate of the
lossy technique. 
\bibliographystyle{splnc}
\bibliography{bitlevel}

\begin{thebibliography}{10}
\providecommand{\url}[1]{\texttt{#1}}
\providecommand{\urlprefix}{URL }
\providecommand{\doi}[1]{https://doi.org/#1}

\bibitem{IEEE754}
ANSI/IEEE: IEEE Standard for Binary Floating-point Arithmetic, std 754-2008
  edn. (2008)

\bibitem{arc}
Bailey, D.: Resolving numerical anomalies in scientific computation  (03 2008)

\bibitem{KM19}
{Ben Khalifa}, D., Martel, M.: Precision tuning and internet of things. In:
  International Conference on Internet of Things, Embedded Systems and
  Communications, {IINTEC} 2019. pp. 80--85. {IEEE} (2019)

\bibitem{KM20}
{Ben Khalifa}, D., Martel, M.: Precision tuning of an accelerometer-based
  pedometer algorithm for iot devices. In: International Conference on Internet
  of Things and Intelligence System, {IOTAIS} 2020. pp. 113--119. {IEEE} (2020)

\bibitem{KMA19}
{Ben Khalifa}, D., Martel, M., Adj{\'{e}}, A.: {POP:} {A} tuning assistant for
  mixed-precision floating-point computations. In: Formal Techniques for
  Safety-Critical Systems - 7th International Workshop, {FTSCS} 2019.
  Communications in Computer and Information Science, vol.~1165, pp. 77--94.
  Springer (2019)

\bibitem{CBBSGR17}
Chiang, W., Baranowski, M., Briggs, I., Solovyev, A., Gopalakrishnan, G.,
  Rakamaric, Z.: Rigorous floating-point mixed-precision tuning. In:
  Proceedings of the 44th {ACM} {SIGPLAN} Symposium on Principles of
  Programming Languages, {POPL}. pp. 300--315. {ACM} (2017)

\bibitem{CGGMP05}
Costan, A., Gaubert, S., Goubault, E., Martel, M., Putot, S.: A policy
  iteration algorithm for computing fixed points in static analysis of
  programs. In: Computer Aided Verification, 17th International Conference,
  {CAV} 2005, Edinburgh, Scotland, UK, July 6-10, 2005, Proceedings. Lecture
  Notes in Computer Science, vol.~3576, pp. 462--475. Springer (2005)

\bibitem{DHS18}
Darulova, E., Horn, E., Sharma, S.: Sound mixed-precision optimization with
  rewriting. In: Proceedings of the 9th {ACM/IEEE} International Conference on
  Cyber-Physical Systems, {ICCPS}. pp. 208--219. {IEEE} Computer Society /
  {ACM} (2018)

\bibitem{DFHS19}
Diffenderfer, J., Fox, A., Hittinger, J.A.F., Sanders, G., Lindstrom, P.G.:
  Error analysis of {ZFP} compression for floating-point data. {SIAM} J. Sci.
  Comput.  \textbf{41}(3),  A1867--A1898 (2019)

\bibitem{GR18}
Guo, H., Rubio{-}Gonz{\'{a}}lez, C.: Exploiting community structure for
  floating-point precision tuning. In: Proceedings of the 27th {ACM} {SIGSOFT}
  International Symposium on Software Testing and Analysis, {ISSTA} 2018. pp.
  333--343. {ACM} (2018)

\bibitem{KSWLB19}
Kotipalli, P.V., Singh, R., Wood, P., Laguna, I., Bagchi, S.: {AMPT-GA:}
  automatic mixed precision floating point tuning for {GPU} applications. In:
  Proceedings of the {ACM} International Conference on Supercomputing, {ICS}.
  pp. 160--170. {ACM} (2019)

\bibitem{LHSL13}
Lam, M.O., Hollingsworth, J.K., de~Supinski, B.R., LeGendre, M.P.:
  Automatically adapting programs for mixed-precision floating-point
  computation. In: International Conference on Supercomputing, ICS'13. pp.
  369--378. {ACM} (2013)

\bibitem{simpsons}
McKeeman, W.M.: Algorithm 145: Adaptive numerical integration by simpson's
  rule. Commun. ACM  \textbf{5}(12), ~604 (1962)

\bibitem{pedometer}
Morris, D., Saponas, T., Guillory, A., Kelner, I.: Recofit: Using a wearable
  sensor to find, recognize, and count repetitive exercises. Conference on
  Human Factors in Computing Systems - Proceedings  (2014)

\bibitem{MB08}
de~Moura, L.M., Bj{\o}rner, N.: {Z3:} an efficient {SMT} solver. In: Tools and
  Algorithms for the Construction and Analysis of Systems. LNCS, vol.~4963, pp.
  337--340. Springer (2008)

\bibitem{papadimitriou1981complexity}
Papadimitriou, C.H.: On the complexity of integer programming. Journal of the
  ACM (JACM)  \textbf{28}(4),  765--768 (1981)

\bibitem{parker}
Parker, D.S.: Monte carlo arithmetic: exploiting randomness in floating-point
  arithmetic. Tech. Rep. CSD-970002, University of California (Los Angeles)
  (1997)

\bibitem{RGNNDKSBIH13}
Rubio{-}Gonz{\'{a}}lez, C., Nguyen, C., Nguyen, H.D., Demmel, J., Kahan, W.,
  Sen, K., Bailey, D.H., Iancu, C., Hough, D.: Precimonious: tuning assistant
  for floating-point precision. In: International Conference for High
  Performance Computing, Networking, Storage and Analysis, SC'13. pp.
  27:1--27:12. {ACM} (2013)

\bibitem{schrijver1998theory}
Schrijver, A.: Theory of linear and integer programming. John Wiley \& Sons
  (1998)

\end{thebibliography}

\newpage

\section{Appendix} \label{sec8}
We need a lemma on some algebraic operations stable on the set of functions written as the min-max of a finite family of affine functions. The functions are defined on $\rr^d$.
\begin{lemma}
\label{lemmaapx}
The following statements hold:
\begin{itemize}
\item The sum of two min-max of a finite family of affine functions is a min-max of a finite family of affine functions.
\item The maximum of two min-max of a finite family of affine functions is a min-max of a finite family of affine functions.
\end{itemize}
\end{lemma}

\begin{proof}
Let $g$ and $h$ be two min-max of a finite family of affine functions and $f=g+h$. We have $g=\min_i\max_j g^{ij}$ and $h=\min_k\max_l h^{kl}$. Let $x\in \rr^d$. There exist $i,k$ such that $f(x)\geq \max_j g^{ij}(x)+\max_l h^{kl}(x)=\max_{j,l} g^{ij}(x)+h^{kl}(x)$. We have also, for all $i,k$, $f(x)\leq \max_j g^{ij}(x)+\max_l h^{kl}(x)=\max_{j,l} g^{ij}(x)+h^{kl}(x)$. We conclude that $f(x)=\min_{i,k}\max_{j,l} g^{ij}(x)+h^{kl}(x)$ for all $x$.
We use the same argument for the max. \hfill$\Box$
\end{proof}

\begin{proposition}
The following results hold:
\begin{enumerate}
\item Let $\xi$ the constant function equal to 1. The system $S_\xi^1$ can be rewritten as $\{\accd\in\nn^{Lab}\mid F(\accd)\leq (\accd)\}$ where $F$ maps $\rr^{Lab}\times \rr^{Lab}$ to itself, $F(\nn^{Lab}\times \nn^{Lab})\subseteq (\nn^{Lab}\times \nn^{Lab})$ and has coordinates which are the maximum of a finite family of affine order-preserving functions.
\item Let $\xi$ the function such that $\xi(\ell)$ equals the function defined at Fig.~\ref{ip}. The system $S_\xi$ can be rewritten as $\{(\accd,\acce)\in\nn^{Lab}\times \nn^{Lab}\mid F(\accd,\acce)\leq (\accd,\acce) \}$ where $F$ maps $\rr^{Lab}\times \rr^{Lab}$ to itself, $F(\nn^{Lab}\times \nn^{Lab})\subseteq (\nn^{Lab}\times \nn^{Lab})$ and all its coordinates are the min-max of a finite family of affine functions.
\end{enumerate}
\end{proposition}

\begin{proof}
We only give details about the system $S_\xi^1$ (Figure~\ref{ilp}).
By induction on the rules. We write $L=\{\ell\in Lab\mid F_\ell \text{ is  constructed }\}$. This set is used in the proof to construct $F$ inductively.

For the rule (CONST), there is nothing to do. For the rule (ID), if the label $\ell'=\rho(x)\in L$ then we define $F_{\ell'}(\accd)=\max(F_{\ell'}(\accd),\accd(\ell))$. Otherwise  $F_{\ell'}(\accd)=\accd(\ell)$. As $\accd\mapsto \accd(\ell)$ is order-preseving and the maximum of one affine function, $F_{\ell'}$ is the maximum of a finite family of order-preserving affine functions since $\max$ preserves order-preservation.

For the rules (ADD), (SUB), (MULT), (DIV), (MATH) and (ASSIGN), by induction, it suffices to focus on the new set of inequalities. If $\ell_1\in L$, we define $F_{\ell_1}$ as the max with old definition and $RHS(\accd)$ i.e. $F_{\ell_1}(\accd)=\max(RHS(\accd),F_{\ell_1}(\accd))$ where $RHS(\accd)$ is the right-hand side part of the new inequality. If $\ell_1\notin L$, we define $F_{\ell_1}(\accd)=RHS(\accd)$. In the latter rules, $RHS(\accd)$ are order-preserving affine functions. It follows that $F_{\ell_1}$ is the maximum of a finite family of order-preserving affine functions.

The result follows by induction for the rule (SEQ).

The rules (COND) and (WHILE) are treated as the rules (ADD), (SUB), (MULT), (DIV), (MATH) and (ASSIGN), by induction and the consideration of the new set of inequalities.

The last rule (REQ) constructs $F_{\rho(x)}$ either as the constant function equal to $p$ at label $\rho(x)$ or the maximum of the old definition of $F_{\rho(x)}$ and $p$ if $\rho(x)\in L$.

The proof for the system  $S_\xi$ uses the same arguments and Lemma~\ref{lemmaapx}. \hfill$\Box$
\end{proof}
\end{document}